\documentclass[11pt]{article}
\pdfoutput=1
\usepackage[centertags]{amsmath}
\usepackage[square, comma, sort&compress,numbers]{natbib}
\usepackage{array,multirow}
\numberwithin{equation}{section}
\usepackage{amssymb,amsfonts,amsthm}
\usepackage{graphicx}
\usepackage{color}
\usepackage{mathtools,bm}
\usepackage{accents}
\usepackage[dvipsnames]{xcolor}
\usepackage{oplotsymbl}
\usepackage{authblk}
\usepackage{mathrsfs,wasysym}

\usepackage{epsfig}
\usepackage{bbold}

\usepackage{wrapfig}
\usepackage{float}
\usepackage{soul}
\usepackage{tkz-euclide}
\usepackage{braket}

\usepackage{tikz,pgf}
\usetikzlibrary{shapes}
\usetikzlibrary{calc}
\usetikzlibrary{decorations.pathmorphing}
\usetikzlibrary{decorations.pathreplacing,shapes.misc}
\usetikzlibrary{positioning}
\usetikzlibrary{arrows}
\usetikzlibrary{decorations.markings}
\usetikzlibrary{shadings}

\usepackage{tikz}
\usetikzlibrary{shapes.geometric, calc}

\usetikzlibrary{intersections}

\def\be{\begin{equation}}
\def\ee{\end{equation}}
\def\ba{\begin{array}}
\def\ea{\end{array}}

\def\bst{\begin{split}}
\def\est{\end{split}}

\def\dps{\displaystyle}

\def\Li2{\operatorname{Li_2}}

\newtheorem{theorem}{Theorem}[section]

\newdimen\tableauside\tableauside=1.0ex
\newdimen\tableaurule\tableaurule=0.4pt
\newdimen\tableaustep
\def\phantomhrule#1{\hbox{\vbox to0pt{\hrule height\tableaurule
width#1\vss}}}
\def\phantomvrule#1{\vbox{\hbox to0pt{\vrule width\tableaurule
height#1\hss}}}
\def\sqr{\vbox{%
\phantomhrule\tableaustep

\hbox{\phantomvrule\tableaustep\kern\tableaustep\phantomvrule\tableaustep}%
\hbox{\vbox{\phantomhrule\tableauside}\kern-\tableaurule}}}
\def\squares#1{\hbox{\count0=#1\noindent\loop\sqr
\advance\count0 by-1 \ifnum\count0>0\repeat}}
\def\tableau#1{\vcenter{\offinterlineskip
\tableaustep=\tableauside\advance\tableaustep by-\tableaurule
\kern\normallineskip\hbox
{\kern\normallineskip\vbox
{\gettableau#1 0 }%
\kern\normallineskip\kern\tableaurule}%
\kern\normallineskip\kern\tableaurule}}
\def\gettableau#1 {\ifnum#1=0\let\next=\null\else
\squares{#1}\let\next=\gettableau\fi\next}

\newtheorem{prop}{Proposition}[section]
\newtheorem{lemma}[prop]{Lemma}
\newtheorem{definition}[prop]{Definition}
\newtheorem{rem}{Remark}[section]

\newtheorem{cor}{Corollary}[section]

\newcommand{\bref}[1]{\textbf{\ref{#1}}}

\newcommand{\RR}{\mathbb{R}}

\newcommand{\ZZ}{\mathbb{Z}}
\newcommand{\NN}{\mathbb{N}}

\def\cA{\mathcal{A}}
\def\cB{\mathcal{B}}

\def\cG{\mathcal{G}}

\def\rT{{\rm T}}
\def\rU{{\rm U}}
\def\rV{{\rm V}}
\def\rW{{\rm W}}

\def\rY{{\rm Y}}

\numberwithin{equation}{section} \makeatletter
\@addtoreset{equation}{section}

\def\be{\begin{equation}}
\def\ee{\end{equation}}
\def\ba{\begin{array}}
\def\ea{\end{array}}

\def\dps{\displaystyle}

\def\C2{\text{C}_2}

\def\HG{{\rm H}}
\def\BF{\Phi}

\newcommand{\triple}[1]{\langle #1 \rangle}

\newcommand{\cham}{{\mathbb{C}}}
\newcommand{\xr}{\textsf{q}}

\usepackage{jheppub}
\makeatletter
\def\@fpheader{\vspace{-.1cm}}
\makeatother

\title{\centering{Diagrammatic construction of GKZ systems \\for polygonal functions}}

\author[a,b]{Konstantin Alkalaev}
\author[a]{\;\;Semyon Mandrygin}
\author[b]{\;\;Yakov Zalishchansky}
\affiliation[a]{I.E. Tamm Department of Theoretical Physics, \\P.N. Lebedev Physical
Institute, 119991 Moscow, Russia}
\affiliation[b]{Institute for Theoretical and Mathematical Physics, \\
Lomonosov Moscow State University, 119991 Moscow, Russia}
\emailAdd{alkalaev@lpi.ru, semyon.mandrygin@gmail.com, eldeeclaud@gmail.com}

\abstract{We investigate the recently introduced polygonal hypergeometric functions \cite{Alkalaev:2025fgn,Alkalaev:2025zhg}, which are conjectured to evaluate multipoint parametric one-loop conformal integrals in arbitrary dimensions. These functions can be systematically constructed using a diagrammatic algorithm that operates in terms of simple planar shapes. Given an $n$-point polygonal power series, we formulate  the corresponding GKZ hypergeometric system, which admits an explicit diagrammatic interpretation.  In particular, the transianic matrix, which encodes all the relevant diagrammatic data, determines the corresponding toric matrix along with a basis of the lattice of relations, thereby singling out a finite subsystem of $n(n-3)/2$ toric equations. This finite subsystem consists of second- and third-order equations, whereas the full toric family also includes higher-order ones. We illustrate the general procedure with several prominent examples of polygonal functions, including the fourth Appell and the Srivastava--Daoust functions.
}

\begin{document}

\maketitle
\flushbottom

\section{Introduction}

Conformal integrals \cite{Symanzik:1972wj} arise in several areas of quantum field theory. Within the shadow formalism \cite{Ferrara:1972xe,Ferrara:1972ay,Ferrara:1972uq,Ferrara:1972kab}, one-loop integrals enter the calculation of conformal partial waves and conformal blocks in $D$-dimensional conformal field theory (CFT$_D$), both in flat space \cite{Dolan:2000ut,Dolan:2011dv,SimmonsDuffin:2012uy,Rosenhaus:2018zqn} and on thermal backgrounds \cite{Alkalaev:2024jxh,Alkalaev:2026sha}.
A related application arises from the correspondence between ladder integrals and thermal observables in the theory of a free massive complex scalar field \cite{Petkou:2021zhg,Karydas:2023ufs,Karydas:2025tfs}.

In quantum field theories with dual conformal symmetry, notably planar $\mathcal N=4$ SYM theory, conformal integrals enter the calculation of scattering amplitudes \cite{Drummond:2006rz,Drummond:2007aua,Drummond:2008vq}. In fishnet theories \cite{Gurdogan:2015csr,Grabner:2017pgm,Caetano:2016ydc,Kazakov:2018qbr}, the corresponding Feynman integrals possess Yangian symmetry in addition to conformal invariance \cite{Chicherin:2017cns,Chicherin:2017frs,Loebbert:2019vcj,Loebbert:2020glj,Corcoran:2021gda,Duhr:2022pch}. This additional symmetry makes various integrability-based approaches available for the study of fishnet integrals \cite{Derkachov:2020zvv,Olivucci:2021cfy,Derkachov:2021ufp,Olivucci:2023tnw,Aprile:2023gnh,Duhr:2024hjf}.
Prominent examples include the Basso--Dixon integrals \cite{Basso:2017jwq,Derkachov:2018rot,Duhr:2023eld} and conformal ladder integrals \cite{Davydychev:1992eww,Derkachov:2022ytx,Derkachov:2023xqq,Loebbert:2024fsj,Derkachov:2025xmt}. This recent progress also motivated the study of more general Feynman integrals with extended symmetries \cite{Loebbert:2020hxk,Loebbert:2020tje,Corcoran:2020epz,Loebbert:2022nfu,Kazakov:2023nyu,Loebbert:2024qbw,Loebbert:2025abz,Ferrando:2025duw}, continuing and developing an earlier line of research \cite{Zamolodchikov:1980mb,Isaev:2003tk}. Other fruitful approaches to conformal integrals include their geometric interpretation as volumes of simplices \cite{Davydychev:1997wa,Mason:2010pg,Schnetz:2010pd,Nandan:2013ip,Bourjaily:2019exo,Ren:2023tuj} and their representation in terms of Mellin--Barnes integrals \cite{Ananthanarayan:2020ncn,Banik:2023rrz}.

A central goal in the study of Feynman integrals is to express them in terms of special functions, such as hypergeometric functions and multiple polylogarithms \cite{Kalmykov:2008gq, Bourjaily:2022bwx}.\footnote{Nevertheless, the analysis of some processes in quantum field theory requires only the asymptotic behavior of the relevant integrals in a given kinematic regime, see e.g. \cite{Bercini:2024pya,Bork:2025ztu}.} These two classes of functions typically arise when evaluating parametric and non-parametric conformal integrals, respectively. In addition to the external coordinates, interpreted as the positions of primary operators or as dual momenta, a parametric integral depends on propagator powers, which are continuous parameters subject to linear constraints required by conformal invariance. Non-parametric integrals are obtained by specializing these parameters to particular physical values. Although parametric integrals  constitute a more general class, calculations at fixed physical values can often be performed directly without first evaluating the integral for arbitrary parameters.

In this paper, we focus on a general setup and study  multipoint one-loop parametric conformal integrals in arbitrary spacetime dimensions, continuing the line initiated in \cite{Alkalaev:2025fgn,Alkalaev:2025zhg}. In these works, an $n$-point conformal integral is decomposed into a finite set of basis functions defined on a specific coordinate domain and related by the action of the cyclic group. Each basis function is given by the product of a normalization constant, a conformally covariant leg-factor, and a hypergeometric-type series referred to as a polygonal function. All three ingredients are systematically constructed via the diagrammatic algorithm based on the analysis of planar figures. In particular, the arguments of the polygonal function are cross-ratios represented by quadrilaterals and pentagons inscribed in an $n$-gon, which can be regarded as part of the Baxter lattice. The relative positions of these figures are characterized by the transianic indices, which determine the coefficients of the polygonal power series.

At present, the origin of this geometric construction remains elusive. The primary challenge lies in providing a systematic explanation for its emergence and clarifying its connection to other established frameworks for describing generalized hypergeometric functions. The general theory developed by Gelfand, Graev, Kapranov, Retakh, and Zelevinsky (GKZ) \cite{Gelfand1992General,GelGraZel87,Gelfand1988,Gelfand1989,Gelfand1990} provides such a systematic framework, and the present work aims to describe polygonal functions within this approach.\footnote{In a broader context, the relationship between GKZ systems and Feynman integrals (including conformal ones) has been extensively studied, see e.g. \cite{Klausen:2019hrg,klausen2023hypergeometricfeynmanintegrals,delaCruz:2019skx, Weinzierl:2022eaz,Pal:2021llg,Pal:2023kgu,Levkovich-Maslyuk:2024zdy}. By providing systems of differential equations, the GKZ approach is also naturally connected to the Yangian bootstrap \cite{Chicherin:2017cns,Chicherin:2017frs,Loebbert:2019vcj}.} In particular, we develop a general procedure for reconstructing the GKZ hypergeometric system from the $n$-point polygonal function for any $n$. This approach bypasses the need for manual, case-by-case derivations, providing instead a universal procedure for deriving the PDEs.

It turns out that the diagrammatic algorithm extends from the construction of polygonal functions to their associated GKZ systems, so that the relevant differential equations can be recovered directly from the underlying geometric data. Our main claim is that the transianic matrix, which encodes the relative configuration of the planar figures entering the diagrammatic algorithm, determines the nontrivial block of the associated toric matrix. At the same time, it determines a basis of the lattice of linear relations among the columns of the toric matrix and thereby organizes the corresponding infinite family of toric equations. In particular, the diagrammatic algorithm singles out a finite subsystem of $n(n-3)/2$ toric equations, one for each cross-ratio: the squared distances in its numerator and denominator determine the two differential monomials of the corresponding toric equation. Finally, we show that the parameters entering the $n$ Euler equations are linearly related to those of the polygonal function and, consequently, to the propagator powers of the original conformal integral.

The paper is organized as follows.  In section \bref{sec:diag}, we reconsider the diagrammatic algorithm. In particular, we establish several properties of the cross-ratio set and the transianic matrix that are essential for the construction of the associated GKZ systems. The need to review GKZ systems motivated our definition-theorem  style of presentation and prompted us to formalize the diagrammatic algorithm. In section \bref{sec:GKZ}, we first outline the basic ingredients of the GKZ approach and then formulate our general procedure for obtaining a GKZ hypergeometric system from a given polygonal function. Specifically, in section \bref{sec:GKZ_polygon} we present toric and Euler PDEs satisfied by  $n$-point polygonal functions. In the concluding section \bref{sec:conclusion} we summarize our construction and discuss open problems. Appendix \bref{app:trans_rels} contains proofs of the main-text lemmas and several new lemmas regarding the transianic indices.  In Appendix  \bref{sec:examples}, we examine a few  basic examples of polygonal functions, including known functions that arise in the context of evaluating conformal integrals. These include  the  $n=4$ polygonal function (the fourth Appell function) and certain  $n=5, 6$ polygonal functions (the Srivastava-Daoust functions). Appendix \bref{app:adjacency} presents the transianic matrix in terms of the adjacency matrices of the graphs associated with the planar figures involved in the diagrammatic algorithm.

\paragraph{Notation and conventions.} $\mathbb{Z}$ denotes the set of integers; $\mathbb{Z}^N$ denotes the $N$-dimensional integer lattice; $\mathbb{Z}_{+}$ denotes non-negative integers; $\mathbb{Z}_{+}^{N}$ denotes the $N$-dimensional non-negative integer lattice;  $\NN_n =\{1,2,...,n\}$; $l\notin\{i,j,k\}$ means that $l\in \mathbb{N}_n\setminus\{i,j,k\}$.

\section{Diagrammatic algorithm}
\label{sec:diag}

This section provides the diagrammatic framework  required to construct the generalized hypergeometric series associated with planar shapes—specifically, the polygonal functions \cite{Alkalaev:2025zhg}. Operating on planar figures, this diagrammatic algorithm systematically determines both the arguments and the expansion coefficients of the multivariate power series. To bridge this geometric construction with the differential systems, we also establish several key propositions that will be essential for formulating the corresponding GKZ hypergeometric systems in  section \bref{sec:GKZ}.

\subsection{Geometric shapes on the plane}
\label{sec:geom}

The central object of our construction is a planar $n$-gon, hereafter referred to as the conformal polygon $P_n$.\footnote{The polygon $P_n$ can be considered as an elementary cell of the Baxter lattice, where the parameters $a_i$ are linearly related to the angles of the polygon. This is closely related to  integrable fishnet graphs \cite{Zamolodchikov:1980mb, Kazakov:2022dbd, Kazakov:2023nyu, Levkovich-Maslyuk:2024zdy}. While a detailed account of the conformal polygon is provided in \cite{Alkalaev:2025zhg}, for our present purposes it suffices to regard it simply as an  $n$-gon in the plane.} Each vertex of the polygon, indexed by $l\in\NN_n$, is assigned a pair $(a_l,x_l)$, where $a_l \in \mathbb{R}$ is a parameter and $x_l \in \mathbb{R}^D$ is a point in  Euclidean coordinate space. For illustrative purposes, we depict $P_n$ as a regular polygon, as shown in fig.~\bref{fig:Conf.pol} {\bf (a)}.

\begin{figure}
\centering
\includegraphics[width=0.8\linewidth]{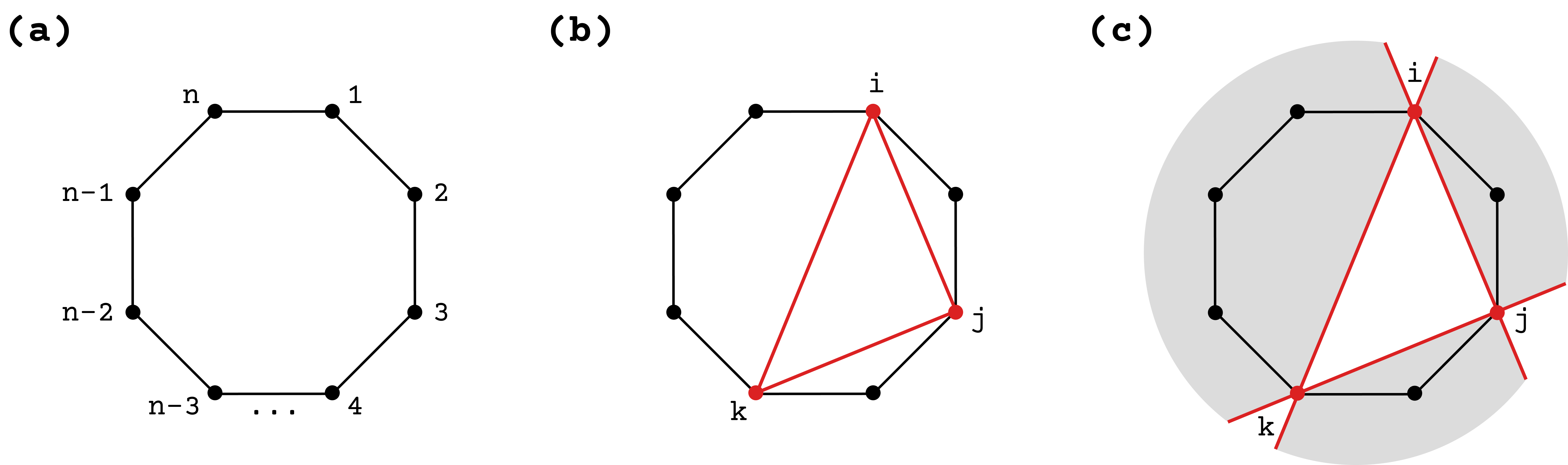}
\caption{{\bf (a)}: The conformal polygon $P_n$ depicted  as a regular $n$-gon. {\bf (b)}: A basis triangle $\triangle_n^{\triple{ijk}}$. {\bf (c)}: The three open chambers (shaded regions) determined by the lines supporting the edges of the basis triangle.}
\label{fig:Conf.pol}
\end{figure}

\begin{definition}
\label{def:dist}
Let  $\bm x = \big\{x_i \in \mathbb{R}^D\mid i\in \NN_n\big\}$ be a set of $n$ points in   the coordinate space, and let $\bm X_n = \big\{X_{ij} = (x_i - x_j)^2 \mid i,j \in \NN_n,\,  i \neq j\big\}$  be the corresponding  set of squared distances.
\end{definition}
\begin{definition}
\label{def:powers}
Let $\bm a = \{a_1, ..., a_n\} \in \mathbb{R}^n$ be a vector in the parameter space that lies on the hyperplane:
\be
\label{conf}
\sum_{i=1}^n a_i = D\,.
\ee
\end{definition}
\noindent Geometrically, imposing the constraint \eqref{conf} implies that the angles of $P_{n}$, which are linearly related to the parameters $a_{i}$, sum up to $\pi(n-2)$. On the other hand, in physical applications,  it arises, for example, from requiring the conformal $O(D+1,1)$ covariance of the corresponding Feynman integral, with the parameters $a_{i}$ interpreted as the propagator powers.
\begin{definition}
\label{def:triples}
Let $\mathrm{R}_n$ denote the set of ordered index triples:
\be
\label{triplet}
\mathrm{R}_n = \big\{\triple{ijk} \mid i,j,k \in \NN_n,\, i < j < k \big\}\,.
\ee
\end{definition}
\begin{definition}
Any three vertices $\{i,j,k\}$ of the polygon $P_n$ define an inscribed triangle, denoted by $\triangle_{n}^{\triple{ijk}}$ and  called a basis triangle.
\end{definition}
\noindent In other words, each ordered triple of indices $\triple{ijk} \in \mathrm{R}_n$ uniquely determines a basis triangle $\triangle_{n}^{\triple{ijk}} \subset P_n$, as illustrated in fig.~\bref{fig:Conf.pol} {\bf (b)}.
\begin{definition}
The straight lines supporting the edges of $\triangle_n^{\triple{ijk}}$ determine  three open regions in the plane that are adjacent to its sides and lie outside the basis triangle. These regions are called open chambers. The chamber adjacent to the edge $(j,k)$, or equivalently opposite to the vertex $i$, is denoted by $\cham_n^{(i)}$.
\end{definition}
\noindent  The chambers are illustrated in fig.~\bref{fig:Conf.pol} {\bf (c)}.

For a  basis triangle $\triangle_{n}^{\triple{ijk}}\subset P_n$, the numbers of vertices of the conformal polygon  that lie in the respective chambers are given by
\be
\label{L_in_vertices}
L_n^{(i)}= k - j -1\,,
\qquad
L_n^{(j)}= i -k - 1  + n  \,,
\qquad
L_n^{(k)}= j - i -1 \,,
\ee
and satisfy the balance relation
\be
\label{balance_relation}
L_n^{(i)} + L_n^{(j)} + L_n^{(k)} = n-3\,.
\ee

\begin{rem}
By a slight abuse of notation, we use the same symbol $\cham_n^{(a)}$, $a\in \{i,j,k\}$, to denote both the  open region in the plane and the subset of vertices of the conformal polygon lying therein. Accordingly, the notation $p\in\cham_n^{(a)}$ simply means that the corresponding vertex  lies in this  chamber.
\end{rem}

\vspace{-1mm}

\noindent In particular, a chamber $\cham_n^{(a)}$ may be empty, in the sense that $L_n^{(a)}=0$.

\begin{figure}
\centering
\includegraphics[width=0.5\linewidth]{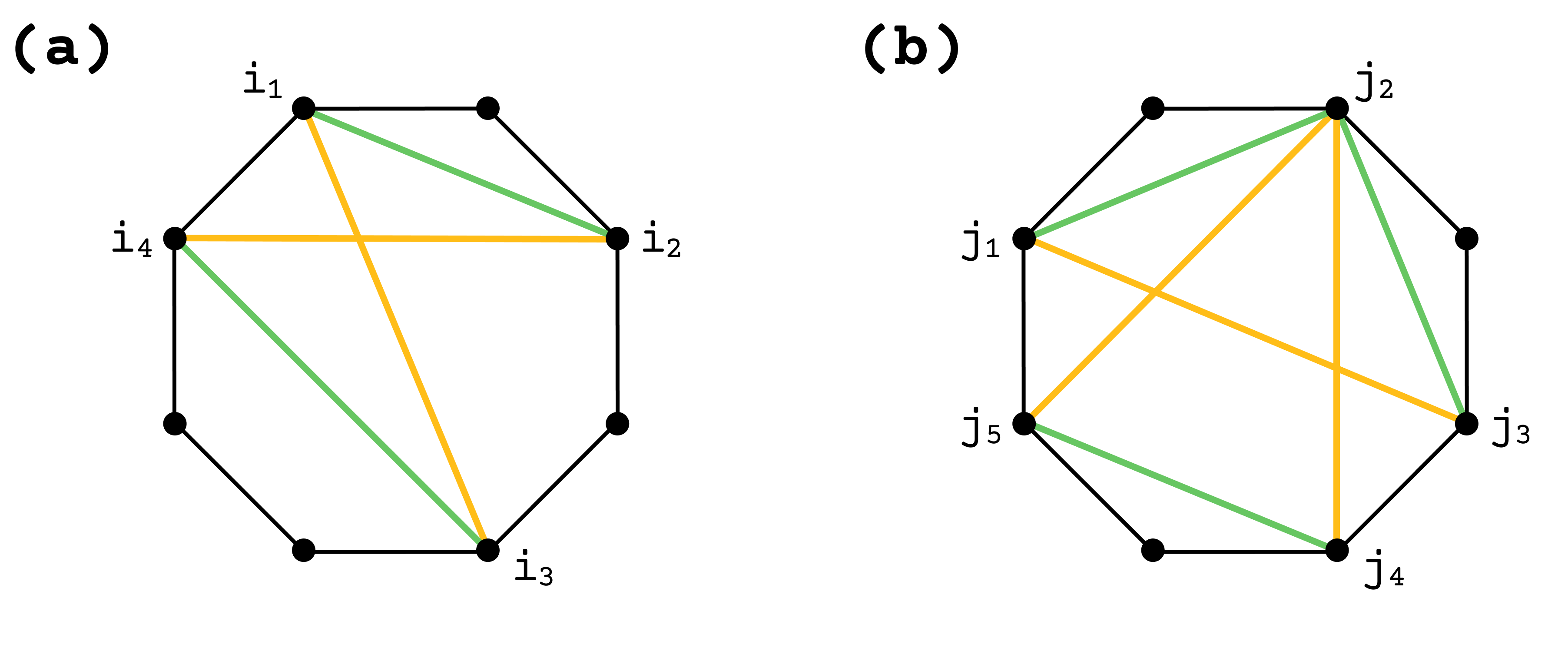}
\caption{Cross-ratio diagrams (Definition \bref{def:green&orange}): {\bf (a)} the quadrilateral with vertices $\{i_1, ..., i_4\}$, and {\bf (b)} the pentagon with vertices $\{j_1, ..., j_5\}$. }
\label{fig:Conf.pol.2}
\end{figure}
\begin{definition}
\label{def:green&orange}
A cross-ratio diagram is a self-intersecting polygon inscribed in $P_n$ whose sides are chords connecting a subset of marked vertices. The chords forming the diagram are colored green and orange according to the following geometric configurations:

\vspace{-3mm}

\begin{itemize}

\item[\textbf{(a)}] \textbf{Quadrilateral  ($n \geqslant 4$):} two opposite edges of the quadrilateral are green, while its two diagonals are orange.

\vspace{-3mm}

\item[\textbf{(b)}] \textbf{Pentagon  ($n \geqslant 5$):} two adjacent edges of the pentagon and the edge opposite to their common vertex are  green. The orange chords are three diagonals: two issuing from this common vertex to the endpoints of the opposite green edge, and one connecting the remaining endpoints of the adjacent green edges.
\end{itemize}
\end{definition}

\vspace{-2mm}

\noindent Typical examples of cross-ratio diagrams are illustrated in fig.~\bref{fig:Conf.pol.2}.

Let us enumerate the vertices of these diagrams clockwise as $\{i_1, i_2, i_4, i_3\}$ for the quadrilateral and $\{j_1, j_2, j_4, j_5, j_3\}$ for the pentagon. This convention yields the analytical representation in terms of cross-ratios. Namely, the cross-ratio diagrams uniquely determine rational functions on the point set $\bm x$, where the green chords form the numerators, while the orange ones form the denominators. These functions are invariant under the $O(D+1,1)$ conformal transformations in $\RR^D$ and are given by
\begin{itemize}

\item[\textbf{(a)}] \textbf{Quadrilateral:} the quadratic cross-ratio is expressed as
\be
\label{U_inv}
U[i_1, i_2, i_3, i_4] := \frac{X_{i_1 i_2} X_{i_3 i_4}}{X_{i_1 i_3} X_{i_2 i_4}}\,,
\ee

\item[\textbf{(b)}] \textbf{Pentagon:} the cubic cross-ratio is expressed as
\be
\label{W_inv}
W[j_1, j_2, j_3, j_4, j_5] := \frac{X_{j_1 j_2} X_{j_2 j_3} X_{j_4 j_5} }{ X_{j_1 j_3} X_{j_2 j_4} X_{j_2 j_5} }\,,
\ee
\end{itemize}
where $X_{ij}\in \bm X_n$.

Formally, the two-color coding defines a weight function $w$ on the chords of the diagram, assigning $+1$ to the green chords and $-1$ to the orange ones, so that the cross-ratios \eqref{U_inv} and \eqref{W_inv} are systematically evaluated as the product of the squared distances $X_{ij}$ raised to the power of their respective weights $w_{ij}$. Ultimately, it is the  diagrammatic algorithm that selects a certain set of such cross-ratio diagrams and their associated functions to serve as the variables of the polygonal functions. This reduces to finding the weights $w_{ij}$.

\subsection{Transianic indices }
\label{sec:key}

At this stage, our construction involves several geometric objects: conformal polygons, basis triangles, and cross-ratio diagrams (see figs.~\bref{fig:Conf.pol} and \bref{fig:Conf.pol.2}). It turns out that the polygonal functions are uniquely  determined by the relative positions of these planar figures, which are effectively characterized by the {\it transianic indices} introduced  below.

\begin{definition}
\label{def1}
Let $G=(V,E)$ be a planar graph, where $V$ is the set of vertices and $E$ is the set of edges $(r,s)$ connecting unordered pairs of distinct  vertices $r,s\in V$. There are three types of graphs, corresponding to the three types of figures involved (conformal polygons, basis triangles, cross-ratio diagrams):
\be
\label{graph_polygon}
\ba{l}
\dps
G_{\hexagon} = \left(V_{\hexagon}, E_{\hexagon}\right) = \big(\{1,2,...,n \}\,, \{(1,2),(2,3),...,(n-1,n),(n,1) \}\big)\,,
\vspace{4mm}
\\
\dps
 G_{\triangle} =  \left(V_{\triangle}, E_{\triangle}\right) =  \big( \{i,j,k\}\,, \{ (i,j),(j,k),(k,i) \} \big)\,,
\vspace{4mm}
\\
\dps
G_{\times}  = \left( V_{\times}\,, E_{\times}\right),
\quad
E_{\times} =E_{\textcolor{green}{\bm \times}} \cup  E_{\textcolor{orange}{\bm\times}}\,.
\ea
\ee
\end{definition}
\noindent The cross-ratio graph is two-colored, meaning that  there is a weight function $E_{\times} \to \{-1, +1\}$ that distinguishes the two edge subsets of $E_{\times}$ ($-1$ for orange and $+1$ for green).  For the quadratic and cubic cross-ratios \eqref{U_inv} and \eqref{W_inv}, the corresponding edge sets are given by
\be
\label{cross_ratio_edge_sets}
\ba{r@{\qquad}l}
U[i_1,i_2,i_3,i_4]:
&E_{\textcolor{green}{\bm\times}}=\big\{(i_1,i_2),(i_3,i_4)\big\}\,,
\\
& E_{\textcolor{orange}{\bm\times}}=\big\{(i_1,i_3),(i_2,i_4)\big\}\,,
\vspace{2mm}
\\
W[j_1,j_2,j_3,j_4,j_5]:
&E_{\textcolor{green}{\bm\times}}=\big\{(j_1,j_2),(j_2,j_3),(j_4,j_5)\big\}\,,
\\
& E_{\textcolor{orange}{\bm\times}}=\big\{(j_1,j_3),(j_2,j_4),(j_2,j_5)\big\}\,.
\ea
\ee

Next, we introduce numerical characteristics that encode the relative positions of the graphs $G_{\triangle}$ and $G_{\bm\times}$ on the vertex set of $G_{\hexagon}$.

\begin{definition}
\label{def2}
For every $l \in V_{\hexagon} \setminus V_{\triangle}$ and for every $(r,s) \in E_{\triangle}$, there are  two transianic indices:
\be
\label{transianic}
\ba{l}
b_l =
\begin{cases}
1, \quad \text{if} \quad l \in V_{\times}\,, \\
0, \quad \text{otherwise} \,,
\end{cases}
\qquad \; b_{rs} =
\begin{cases}
\phantom{-} 1, \quad \text{if} \quad (r,s) \in E_{\textcolor{green}{\bm\times}}\,, \\
-1, \quad \text{if} \quad (r,s) \in E_{\textcolor{orange}{\bm\times}}\,, \\
\phantom{-} 0, \quad \text{otherwise} \,,
\end{cases}
\ea
\ee
which are called the first and second transianic indices, respectively.
\end{definition}
\noindent
Any quadratic $U[i_1,...,i_4]$ or cubic $W[j_1,...,j_5]$ cross-ratio is assigned a sequence of $n$ elements of the set $\{-1,0,1\}$, which describes the position of the corresponding cross-ratio diagram with respect to the basis triangle $\triangle_{n}^{\triple{ijk}}$.

\begin{rem}
\label{rem:adjacency}
The transianic indices admit an equivalent description in terms of the adjacency matrices of the graphs $G_{\hexagon}$, $G_{\triangle}$, and $G_{\times}$, all written on the common vertex set $V_{\hexagon}$. The first indices are recovered from the diagonal entries of the squared adjacency matrices, while the second indices are obtained by restricting the signed adjacency matrix of $G_{\times}$ to the edges of $G_{\triangle}$, with the green and orange edges assigned the values $+1$ and $-1$, respectively. The explicit formulas and their derivation are given in Appendix \bref{app:adjacency}.
\end{rem}

\subsection{Chord sets}
\label{sec:chords}

In this section, we introduce a key classification of the chords of the conformal polygon, which, according to Definition \bref{def:green&orange}, are used to construct cross-ratio diagrams.
\begin{definition}
\label{def:all_chords}
Let $\bm C_n$ denote  the set of all chords connecting unordered pairs of distinct vertices in $V_{\hexagon}$:
\be
\label{set_of_all_chords}
\bm C_n = \Big\{(r,s) \mid r,s \in V_{\hexagon}, \, r \neq s \Big\}\,.
\ee
\end{definition}
\noindent
Note that all the edge sets introduced above --- namely,  $E_{\hexagon}$, $E_{\triangle}$, $E_{\textcolor{green}{\bm \times}}$, and $E_{\textcolor{orange}{\bm \times}}$  --- are subsets of $\bm C_n$.

\begin{definition}
\label{def:chords}
For a basis triangle $\triangle_{n}^{\triple{ijk}}\subset P_n$, the set of basis chords is defined as
\be
\label{set_of_basic_chords}
\bm C_n^{\triple{ijk}} = E_{\triangle} \cup \Big\{(a,l) \mid a\in V_{\triangle}\,, l \in \cham_n^{(a)} \Big\}\,.
\ee
The complementary set of non-basis chords is given by
\be
\label{set_of_non_basis_chords}
\overline{\bm C}_n^{\triple{ijk}} = \bm C_n \setminus \bm C_n^{\triple{ijk}}\,.
\ee
\end{definition}
\noindent An illustration of basis and non-basis chords is provided in fig.~\bref{set_of_basics}.

\begin{figure}[!t]
    \centering
    \includegraphics[width=0.6\linewidth]{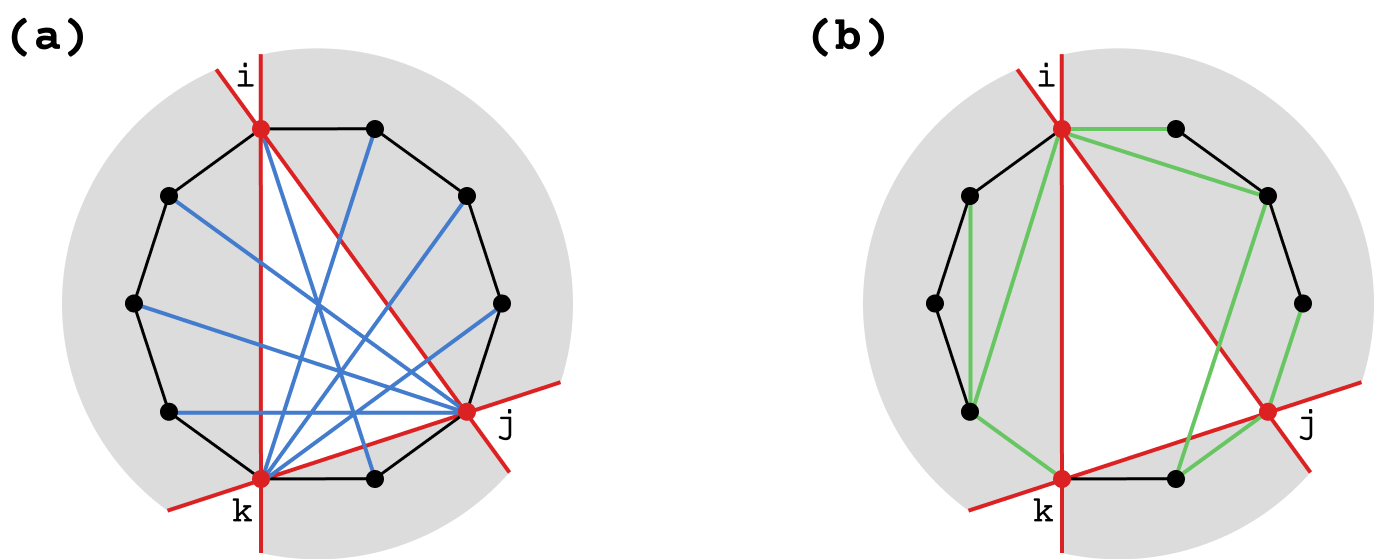}
    \caption{Chords on the conformal polygon. {\bf (a)} Basis chords. The chords corresponding to the edges of the basis triangle are shown in red, while the chords associated with the vertices outside the basis triangle are shown in blue. {\bf (b)} Examples of non-basis chords shown in green.}
    \label{set_of_basics}
\end{figure}

The set $\bm C_n$ contains
\be
\label{all_chords_cardinality}
|\bm C_n| = \binom{n}{2} = \frac{n(n-1)}{2}
\ee
chords. For every $\triple{ijk}\in\mathrm{R}_n$, the corresponding sets of (non-)basis chords have the following cardinalities:
\be
\label{basis_non_basis_cardinalities}
\left|\bm C_n^{\triple{ijk}}\right| = n\,,
\qquad
\left|\overline{\bm C}_n^{\triple{ijk}}\right| = \frac{n(n-3)}{2}\,.
\ee

In summary, specifying a basis triangle $\triangle_{n}^{\triple{ijk}}\subset P_n$  induces a splitting of the chord set into two subsets, which  we refer to as  the chord decomposition:
\be
\label{chord_set_decomposition}
\bm C_n
=
\bm C_n^{\triple{ijk}}
\cup
\overline{\bm C}_n^{\triple{ijk}}\,.
\ee
Furthermore, the chord set  $\bm C_n$ is in a one-to-one correspondence with the distance set $\bm X_n$; hence, the chord decomposition induces a corresponding decomposition of $\bm X_n$.

\subsection{Diagrammatic rules}
\label{sec:diagrammatic_rules}
Given a basis triangle $\triangle_{n}^{\triple{ijk}} \subset P_n$, the diagrammatic algorithm constructs a specific set of $n(n-3)/2$ cross-ratios \eqref{U_inv}--\eqref{W_inv} denoted by $\mathbf{Y}_n^{\triple{ijk}}$, for which we assume the following decomposition:
\be
\label{Y_split}
\mathbf{Y}_n^{\triple{ijk}} = \mathbf{U}_n^{\triple{ijk}} \cup \mathbf{W}_n^{\triple{ijk}} =
{}^{_1}\hspace{-0.5mm}\bm \rU_n^{\triple{ijk}} \cup  {}^{_2}\hspace{-0.5mm}\bm \rU_n^{\triple{ijk}} \cup \mathbf{W}_n^{\triple{ijk}}\,.
\ee
Here, we distinguish between quadratic and cubic cross-ratios, $U[i_1,..., i_4] \in \mathbf{U}_n^{\triple{ijk}}$ and $W[j_1, ..., j_5] \in \mathbf{W}_n^{\triple{ijk}}$, respectively. The procedure described below explains the splitting of the subset of quadratic cross-ratios  $\mathbf{U}_n^{\triple{ijk}} = {}^{_1}\hspace{-0.5mm}\bm \rU_n^{\triple{ijk}} \cup {}^{_2}\hspace{-0.5mm}\bm \rU_n^{\triple{ijk}}$.

\begin{figure}
\centering
\includegraphics[width=0.8\linewidth]{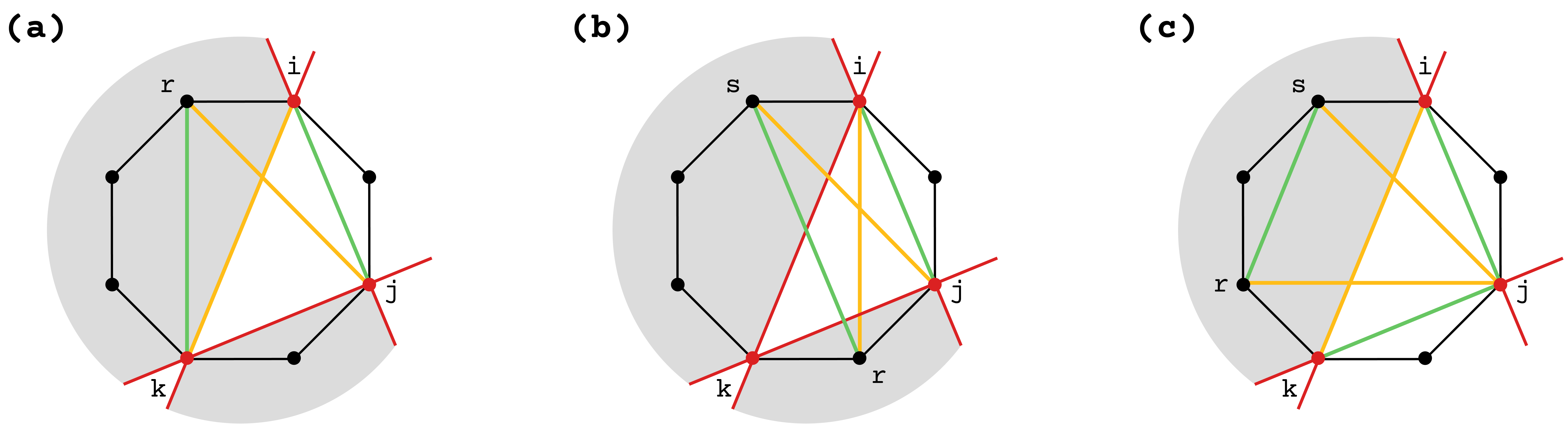}
\caption{{\bf (a)}, {\bf (b)}, {\bf (c)}:  Cross-ratio diagrams  representing elements from ${}^{_1}\hspace{-0.5mm}\bm \rU_n^{\triple{ijk}} \,, {}^{_2}\hspace{-0.5mm}\bm \rU_n^{\triple{ijk}} \,, \mathbf{W}_n^{\triple{ijk}}$, respectively.}
\label{Conf.Pol.11}
\end{figure}

\paragraph{Algorithm.} The cross-ratio set  \eqref{Y_split} is constructed  as follows.

\begin{itemize}

\item[${}^{_1}\hspace{-0.5mm}\bm \rU_n^{\triple{ijk}}$:]
A green chord connecting two vertices of the basis triangle $\triangle_{n}^{\triple{ijk}}$, for instance $i$ and $j$, is considered. A vertex $r$ is chosen within one of the chambers $\cham_n^{(i)}$ or $\cham_n^{(j)}$, and is joined by a green chord to the remaining vertex $k$ of $\triangle_n^{\triple{ijk}}$, as shown in fig.~\bref{Conf.Pol.11} \textbf{(a)}. The green chords determine the numerator, while the denominator is then formed by drawing orange chords along the diagonals of the corresponding quadrilateral  diagram. The resulting cross-ratio, either $U[i,j,r,k]$ or $U[i,j,k,r]$, is an element of the set denoted by ${}^{_1}\hspace{-0.5mm}\bm \rU_n^{\triple{ijk}}$.

\item[${}^{_2}\hspace{-0.5mm}\bm \rU_n^{\triple{ijk}}$:]
A green chord connecting two vertices of the basis triangle $\triangle_{n}^{\triple{ijk}}$, for instance $i$ and $j$, is considered. Vertices $r$ and $s$ are chosen within the open chambers $\cham_n^{(i)}$ and $\cham_n^{(j)}$, respectively, and are joined by a green chord, as shown in fig.~\bref{Conf.Pol.11} \textbf{(b)}. The green chords determine the numerator, while the denominator is formed by drawing orange chords along the diagonals of the corresponding quadrilateral diagram. The resulting cross-ratio $U[i,j,r,s]$ is an element of the set denoted by ${}^{_2}\hspace{-0.5mm}\bm \rU_n^{\triple{ijk}}$.

\item[$\bm \rW_n^{\triple{ijk}}$:]
Green chords connecting two vertices of the basis triangle $\triangle_{n}^{\triple{ijk}}$, for instance $(i,j)$ and $(j,k)$, are considered. Vertices $r$ and $s$ are chosen within the open chamber $\cham_n^{(j)}$ and are joined by a green chord, as shown in fig.~\bref{Conf.Pol.11}\textbf{(c)}. The green chords determine the numerator, while the denominator is formed by drawing orange chords along the diagonals of the corresponding pentagon diagram. The resulting cross-ratio $W[i, j, k, r,s]$ is an element of the set denoted by $\bm \rW_n^{\triple{ijk}}$.
\end{itemize}

By iterating the described procedure over all vertices and edges of $\triangle_n^{\triple{ijk}}$, we obtain the respective cardinalities of these three subsets:
\be
\label{subsets_cardinality}
\ba{l}
\dps
|{}^{_1}\hspace{-0.5mm}\bm \rU_n^{\triple{ijk}}| = 2 n - 6\,,
\quad
|{}^{_2}\hspace{-0.5mm}\bm \rU_n^{\triple{ijk}}| = L_n^{(i)} L_n^{(j)} + L_n^{(i)} L_n^{(k)} + L_n^{(j)} L_n^{(k)}\,,
\vspace{2mm}\\
\dps
|\bm \rW_n^{\triple{ijk}}| = \binom{L_n^{(i)}}{2} + \binom{L_n^{(j)}}{2} + \binom{L_n^{(k)}}{2}\,.
\ea
\ee
Consequently, the cardinality of the total set \eqref{Y_split} constructed by the algorithm is independent of the choice of  $\triangle_{n}^{\triple{ijk}}$:
\be
\label{cardinality}
|\bm \rY_n^{\triple{ijk}}| = |{}^{_1}\hspace{-0.5mm}\bm \rU_n^{\triple{ijk}}| + |{}^{_2}\hspace{-0.5mm}\bm \rU_n^{\triple{ijk}}| + |\bm \rW_n^{\triple{ijk}}| = \frac{n(n-3)}{2}\,.
\ee
This is a direct consequence of the balance relation \eqref{balance_relation}.

The diagrammatic algorithm distinguishes between basis and non-basis chords associated with the basis triangle $\triangle_n^{\triple{ijk}}$ (see Definition \bref{def:chords}), assigning them different roles in the construction of cross-ratios. This property is formalized below.
\begin{prop}
\label{rem:cross_ratio_non_basis_chord}
For a  basis triangle $\triangle_n^{\triple{ijk}}$, every cross-ratio diagram constructed by the algorithm contains a unique non-basis chord, which is always green:
\be
\label{cross_ratio_non_basis_chord}
\ba{c@{\hspace{1.5cm}}c@{\qquad}c}
\text{subset} & \text{cross-ratio} & \text{non-basis chord} \\[1mm]
{}^{_1}\hspace{-0.5mm}\bm \rU_n^{\triple{ijk}} & U[i,j,r,k]\,, \; U[i,j,k,r] & (k,r)
\vspace{1mm}
\\
{}^{_2}\hspace{-0.5mm}\bm \rU_n^{\triple{ijk}} & U[i,j,r,s] & (r,s)
\vspace{1mm}
\\
\bm \rW_n^{\triple{ijk}} & W[i,j,k,r,s] & (r,s)
\ea
\ee
Together with the cyclic permutations of $i,j,k$, these mappings exhaust all non-basis chords associated with the basis triangle, thereby establishing  a bijection:
\be
\label{1-1YC}
\bm \rY_n^{\triple{ijk}}\, \sim \, \overline{\bm C}_n^{\triple{ijk}}\,.
\ee
Moreover, in each cross-ratio diagram, apart from the non-basis chord \eqref{cross_ratio_non_basis_chord}, only the edges of $\triangle_n^{\triple{ijk}}$ may be green. Equivalently, the remaining basis chords $(a,b) \in \bm C_n^{\triple{ijk}} \setminus E_{\triangle}$ are always orange.
\end{prop}

The following proposition completes the first stage of the diagrammatic algorithm --- namely, the construction of the arguments of the polygonal function.
\begin{prop}
\label{prop:arguments}
The set (\ref{Y_split}) is the set of cross-ratio arguments of the $n$-point polygonal function labelled by the index triple $\triple{ijk}\in\mathrm{R}_n$:
\be
\label{DA_cr}
\bm \rY_n^{\triple{ijk}} = \Big\{\rY_p^{\triple{ijk}}\,:\, p = 1,..., \xr\,,\; \text{where}\;\; \xr = \frac{n(n-3)}{2} \Big\}\,,
\ee
with each cross-ratio  $\rY_p^{\triple{ijk}}$ constructed  by the diagrammatic algorithm.
\end{prop}
\begin{rem}
\label{rem:nota}
Each element $\rY_p^{\triple{ijk}} \in \bm \rY_n^{\triple{ijk}}$ is assigned the set of vertices and the set of (colored) edges. The notation in \eqref{graph_polygon} is specified as follows:
\be
G_{\times}^{(p)}  = \left( V_{\times}^{(p)}\,, E_{\times}^{(p)}\right),
\quad
E_{\times}^{(p)} =E_{\textcolor{green}{\bm \times}}^{(p)} \cup  E_{\textcolor{orange}{\bm\times}}^{(p)}\,.
\ee
\end{rem}

\subsection{Transianic matrix }
\label{sec:trans_rels}

Recall that each element of the cross-ratio set $\bm \rY_n^{\triple{ijk}}$ is assigned a sequence of $n$ transianic index values, as specified in Definition \bref{def2}. These collections of numbers are conveniently organized into a matrix.\footnote{Note that the following definition of the transianic matrix differs from that in \cite{Alkalaev:2025zhg} in two respects: the signs of the second transianic indices are reversed, and these indices are placed in the $i$-th, $j$-th, and $k$-th rows.}
\begin{definition}
\label{def_of_trans_mat}
The transianic matrix of the cross-ratio set $\bm \rY_n^{\triple{ijk}}$ is the matrix
$\cB_n^{\triple{ijk}}
= \big(\cB_l^{(p)}\big)_{1\leq l \leq n,\, 1\leq p \leq \xr}
\in \operatorname{Mat}_{n\times\xr}(\mathbb{Z})$ whose entries are defined by
\be
\label{tr_mat}
\cB_l^{(p)}=
\begin{cases}
\phantom{-} b_l^{(p)}\,, & l \notin \{i,j,k\}\,,\\
-b_{ij}^{(p)}\,, & l=i\,,\\
-b_{jk}^{(p)}\,, & l=j\,,\\
-b_{ki}^{(p)}\,, & l=k\,.
\end{cases}
\ee
Here, the column label $(p)$ corresponds to the cross-ratio $\rY_p^{\triple{ijk}}\in\bm \rY_n^{\triple{ijk}}$.
\end{definition}
\noindent
Equivalently, the transianic matrix can be viewed as a  vector configuration:
\be
\label{transianic_vector_configuration}
\cB_n^{\triple{ijk}}
=
\left\{
{\bm \cB}^{(1)},...,{\bm \cB}^{(\xr)}
\;\middle|\;
{\bm \cB}^{(p)}
=
\big(\cB_1^{(p)},...,\cB_n^{(p)}\big)^\rT
\in\mathbb{Z}^n\,,
\quad
p=1,...,\xr
\right\}\,.
\ee
Here, each {\it transianic} vector ${\bm \cB}^{(p)}$ represents the $p$-th column of the  matrix $\cB_n^{\triple{ijk}}$ \eqref{tr_mat} and collects the transianic indices calculated for $\rY_p^{\triple{ijk}}$. This representation will be useful in section \bref{sec:GKZ}, when formulating the GKZ system for polygonal functions.

The cross-ratio set \eqref{DA_cr} is defined by its transianic matrix. Below we establish a few important properties described by linear relations among the transianic indices. In particular, we will use the following lemma, whose proof is given in Appendix \bref{app:trans_rels}.
\begin{lemma}
\label{cor1}
\label{cor3}
For every $\rY_p^{\triple{ijk}}\in\bm\rY_n^{\triple{ijk}}$, the transianic indices satisfy
\begin{subequations}
\label{property1}
\begin{align}
\rY_p^{\triple{ijk}} \in {}^{m}\hspace{-0.5mm}\bm\rU_n^{\triple{ijk}}:
\qquad
\sum_{l\notin\{i,j,k\}}b_l^{(p)}
&=m \,,
\qquad
b_{ij}^{(p)}+b_{jk}^{(p)}+b_{ki}^{(p)}=m-1\,,
\quad m=1,2\,;
\label{property1_U}
\\[2mm]
\rY_p^{\triple{ijk}} \in \bm\rW_n^{\triple{ijk}}:
\qquad
\sum_{l\notin\{i,j,k\}}b_l^{(p)}
&=2 \,,
\qquad
b_{ij}^{(p)}+b_{jk}^{(p)}+b_{ki}^{(p)}=1\,.
\label{property1_W}
\end{align}
\end{subequations}
\end{lemma}
As a consequence, we obtain the following property.
\begin{prop}
\label{prop:FLI}
For every $\rY_p^{\triple{ijk}}\in \bm \rY_n^{\triple{ijk}}$, the components of the corresponding transianic vector satisfy
\be
\label{FLR}
\sum_{l = 1}^{n} \cB_l^{(p)}  = 1 \,,
\qquad
p = 1,..., \xr\,.
\ee
\end{prop}

\noindent The linear identity \eqref{FLR} will encode the key analytic properties of the polygonal functions and play a crucial role in formulating the associated GKZ system.

In addition to \eqref{FLR}, the bijection between cross-ratio variables and non-basis chords established in Proposition \bref{rem:cross_ratio_non_basis_chord} leads to additional linear relations among the transianic indices.
\begin{lemma}
\label{lemma:local_inversion_balance}
Let $(r_p,s_p)$ be the unique non-basis chord associated with $\rY_p^{\triple{ijk}}$. Then,
\be
\label{local_inversion_balance}
-b_{ij}^{(p)}-b_{ki}^{(p)}
+\sum_{l\in\cham_n^{(i)}}b_l^{(p)}
=
\begin{cases}
1\,, & i\in\{r_p,s_p\}\,,\\
0\,, & i\notin\{r_p,s_p\}\,.
\end{cases}
\ee
Analogous relations are obtained by cyclic permutations of the vertices $i,j,k$.
\end{lemma}

There is a unified representation of the cross-ratios in terms of the transianic matrix.

\begin{prop}
\label{prop:diagrammatic_cross_ratio}
Each cross-ratio $\rY_p^{\triple{ijk}}\in \bm \rY_n^{\triple{ijk}}$ is a weighted product of the squared distances $X_{ab}\in \bm X_n$ associated with the basis chords and the squared distance associated with the unique non-basis chord:
\be
\label{diagrammatic_cross_ratio}
\rY_p^{\triple{ijk}}
=
X_{r_p s_p}\,
X_{ij}^{b_{ij}^{(p)}}
X_{jk}^{b_{jk}^{(p)}}
X_{ki}^{b_{ki}^{(p)}}
\prod_{a\in V_{\triangle}\vphantom{\cham_n^{(a)}}}\,
\prod_{l\in\cham_n^{(a)}}
X_{al}^{-b_l^{(p)}}\,,
\ee
where $X_{r_p s_p}$ corresponds to the non-basis chord $(r_p,s_p) \in \overline{\bm C}_n^{\triple{ijk}}$.
\end{prop}
\begin{proof}
Given a cross-ratio diagram, its green and orange chords determine, respectively, the numerator and denominator of the corresponding function of the squared distances:
\be
\label{cross_general_prod}
\rY_p^{\triple{ijk}}
=
\frac{
\dps\prod_{(a,b)\in E^{(p)}_{\textcolor{green}{\bm\times}}}
X_{ab}}
{
\dps\prod_{(c,d)\in E^{(p)}_{\textcolor{orange}{\bm\times}}}
X_{cd}}\,.
\ee
The chords in each product can be split into basis and non-basis ones, according to the chord decomposition \eqref{chord_set_decomposition}. Proposition \bref{rem:cross_ratio_non_basis_chord} states that each cross-ratio diagram contains a unique non-basis chord $(r_p,s_p)$, which is always green.  In turn, the basis-chord set  \eqref{set_of_basic_chords} is decomposed into the edges of the basis triangle $E_{\triangle}$ and the remaining basis chords. For every $(a,b) \in E_{\triangle}$, the second transianic index \eqref{transianic} equals  $1$, $-1$ or $0$, depending on whether the chord is green, orange, or absent, respectively. Consequently, the cross-ratio \eqref{cross_general_prod} takes the form
\be
\rY_p^{\triple{ijk}}
=
X_{r_p s_p}\,
X_{ij}^{b_{ij}^{(p)}}
X_{jk}^{b_{jk}^{(p)}}
X_{ki}^{b_{ki}^{(p)}}
\,
\frac{
\dps\prod_{(a,b)\in\,
E^{(p)}_{\textcolor{green}{\bm\times}}
\cap\,(\bm C_n^{\triple{ijk}}\setminus E_{\triangle})}
X_{ab}}
{
\dps\prod_{(c,d)\in\,
E^{(p)}_{\textcolor{orange}{\bm\times}}
\cap\,(\bm C_n^{\triple{ijk}}\setminus E_{\triangle})}
X_{cd}}\,.
\ee
Proposition \bref{rem:cross_ratio_non_basis_chord} further states that any basis chord outside $E_{\triangle}$ that occurs in the diagram is orange, so $E^{(p)}_{\textcolor{green}{\bm\times}}
\cap\left(\bm C_n^{\triple{ijk}}\setminus E_{\triangle}\right) = \varnothing $. Finally, for $a\in V_{\triangle}$ and $l\in\cham_n^{(a)}$, the first transianic index $b_l^{(p)}$ determines whether the chord $(a,l)$ occurs in the diagram. Thus,
\be
\prod_{(c,d)\in
E^{(p)}_{\textcolor{orange}{\bm\times}}
\cap\left(\bm C_n^{\triple{ijk}}\setminus E_{\triangle}\right)}
X_{cd}^{-1}
=
\prod_{a\in V_{\triangle}\vphantom{\cham_n^{(a)}}}
\prod_{l\in\cham_n^{(a)}}
X_{al}^{-b_l^{(p)}}\,,
\ee
which completes the proof of \eqref{diagrammatic_cross_ratio}.
\end{proof}
\begin{rem}
Proposition \bref{prop:FLI} and Lemma \bref{lemma:local_inversion_balance}, which encode the properties of the transianic indices, provide an independent verification of the  $O(D+1,1)$ invariance of the representation  \eqref{diagrammatic_cross_ratio}.
\end{rem}

\begin{rem}
\label{rem:diagrammatic_interpretation}
Recalling the chord decomposition (\ref{chord_set_decomposition}), we observe that for each
non-basis chord $(r_p,s_p) \in \overline{\bm C}_n^{\triple{ijk}}$, where $p=1,..., \xr$, the diagrammatic algorithm selects a unique subset of basis
chords $\bm S_p^{\triple{ijk}} \subset \bm C_n^{\triple{ijk}}$. The transianic matrix then defines a sign map on $\bm S_p^{\triple{ijk}}$ that determines which squared distances  enter the numerator and which enter the denominator, ensuring that the resulting ratio is conformally invariant.
In this way, the algorithm explicitly realizes the bijection between $\bm \rY_n^{\triple{ijk}}$ and $\overline{\bm C}_n^{\triple{ijk}}$ established in Proposition \bref{rem:cross_ratio_non_basis_chord}.
\end{rem}

In summary, the transianic matrix and the chord decomposition, intertwined within Proposition \bref{prop:diagrammatic_cross_ratio}, constitute a central component of the diagrammatic algorithm.

\subsection{Polygonal functions}
\label{sec:polygonal}

Given an index triple $\triple{ijk}\in \mathrm{R}_n$, the diagrammatic algorithm provides the set of cross-ratios and the corresponding transianic matrix. We then use these data to define the hypergeometric-type function.
\begin{definition}
\label{def:cyclic_parameter_sums}
For any pair of vertices $r,s\in V_{\hexagon}$, we use the notation
\be
\label{aa}
|{\bm a}_{r,s}|=
\begin{cases}
\dps\sum_{l=r}^s a_l\,, & r<s\,,
\\
\dps\sum_{l=r}^n a_l+\sum_{l=1}^s a_l\,, & r>s\,;
\end{cases}
\qquad
|{\bm a}_{r,s}|'=\frac{D}{2}-|{\bm a}_{r,s}|\,.
\ee
\end{definition}
\begin{definition}
\label{def:polygonal_function}
Fixing $\triple{ijk}\in\mathrm{R}_n$,  let $\bm \rY_n^{\triple{ijk}}$  and $\cB_n^{\triple{ijk}}$  be the cross-ratio set and the transianic matrix, respectively. Given the parameter vector  ${\bm a}$ from  Definition \bref{def:powers}, the $n$-point polygonal function is defined as the following hypergeometric series \cite{Alkalaev:2025zhg}:
\be
\label{general_polygonal}
\HG_n^{\triple{ijk}}({\bm a}|\mathbf{Y}_n^{\triple{ijk}})
=\sum_{\bm m \in \mathbb{Z}_{+}^{\xr}}{\rm A}^{\triple{ijk}}({\bm a}|{\bm m})
\prod_{p=1}^{\xr}\frac{\big(\rY_p^{\triple{ijk}}\big)^{m_p}}{m_p!}\,,
\ee
where the series coefficients are given by
\be
\label{bare_n_2}
{\rm A}^{\triple{ijk}}({\bm a}|{\bm m})
=\frac{\dps\prod_{l\notin\{i,j,k\}}(-)^{M_l}(a_l)_{M_l}}
{(1+|{\bm a}_{i,j}|')_{M_{ij}}(1+|{\bm a}_{j,k}|')_{M_{jk}}(1+|{\bm a}_{k,i}|')_{M_{ki}}}\,.
\ee
Here, $(a)_M = \Gamma(a+M)/\Gamma(a)$ is the Pochhammer symbol, and $\Gamma(z)$ is the Gamma function. The summation index ${\bm m}=(m_1,...,m_\xr)\in\mathbb{Z}_{+}^{\xr}$ and the transianic indices determine the following linear combinations of the summation variables:
\be
\label{MM}
M_l=\sum_{p=1}^{\xr}b_l^{(p)}m_p\,,
\qquad
M_{st}=\sum_{p=1}^{\xr}b_{st}^{(p)}m_p\,.
\ee
The superscript $(p)$ enumerates the transianic indices associated with $\rY_p^{\triple{ijk}}\in\bm\rY_n^{\triple{ijk}}$.
\end{definition}

The convergence of the polygonal functions can be effectively  analyzed by employing the framework developed in \cite{Horn1889, exton1976multiple}. To this end, one introduces a \textit{convergence index},  which  in the present context is  expressed in terms of the transianic indices \cite{Alkalaev:2025zhg}.
\begin{definition}
\label{def:conv_ind}
The convergence index associated with  the $p$-th variable of the polygonal function $\HG_n^{\triple{ijk}}({\bm a}|\mathbf{Y}_n^{\triple{ijk}})$ (\ref{general_polygonal}) is given by
\be
\label{convergence_index}
\Delta^{\triple{ijk}}_p = 1 + b_{ij}^{(p)} + b_{jk}^{(p)} + b_{ki}^{(p)} - \sum_{l \notin \{i,j,k\}} b_l^{(p)}\,.
\ee
\end{definition}
By construction, there are $\xr$ such indices, which characterize  the convergence of the polygonal function in the space of variables.  A given power series diverges when the convergence indices are negative, converges everywhere when they are positive, and possesses a non-vanishing radius of convergence $r_p$ ($p=1,..., \xr$), which is determined   by the following relation when these indices vanish:
\be
\label{convergence_rad}
r_p = |F_p(\bm {m})|^{-1}\,,
\qquad
F_p(\bm {m})=\frac{\dps \prod_{l \notin \{i,j,k\}} (M_l)^{b_l^{(p)}}}{m_p(M_{ij})^{b_{ij}^{(p)}}(M_{jk})^{b_{jk}^{(p)}}(M_{ki})^{b_{ki}^{(p)}} }\,.
\ee
Crucially, the convergence index \eqref{convergence_index} coincides with the left-hand side of the  linear identity \eqref{FLR}. Consequently, the transianic matrix completely determines the convergence domain of a given polygonal function. This connection allows us to establish the following key result regarding the analyticity of the polygonal functions.
\begin{prop}
\label{prop:conv}
The polygonal function $\HG_n^{\triple{ijk}}({\bm a}|\mathbf{Y}_n^{\triple{ijk}})$ (\ref{general_polygonal})  admits  a non-empty convergence domain because    its convergence indices (\ref{convergence_index}) vanish identically:
\be
\Delta^{\triple{ijk}}_p = 0\,,
\qquad
p = 1, ..., \xr\,.
\ee
\end{prop}

Next, we  introduce the so-called basis function,\footnote{Both the function and its name were proposed  in \cite{Alkalaev:2025fgn,Alkalaev:2025zhg} in the context of calculating conformal integrals.} which is  a suitable  modification of the polygonal function that naturally arises within the associated GKZ hypergeometric system, see section \bref{sec:GKZ}.
\begin{definition}
\label{def:bas}
Given the polygonal function $\HG_n^{\triple{ijk}}({\bm a}|\mathbf{Y}_n^{\triple{ijk}})$, the basis function is defined as
\be
\label{basisfunc}
\Phi_n^{\triple{ijk}}(\bm a|\bm x) = {\rm S}_n^{\triple{ijk}}(\bm a) \, {\rm V}_n^{\triple{ijk}}(\bm a|\bm x)\, \HG_n^{\triple{ijk}}({\bm a}|\mathbf{Y}_n^{\triple{ijk}})\,,
\ee
where the $\bm x$-independent \textit{triangle-factor} and the $\bm x$-dependent \textit{leg-factor}  are given by
\be
\label{star_n}
{\rm S}_{n}^{\triple{ijk}}({\bm a}) =
\frac{\Gamma(-|\bm a_{i,j}|') \Gamma(-|\bm a_{j,k}|')
\Gamma(-|\bm a_{k,i}|')}
{\Gamma(a_i) \Gamma(a_j) \Gamma(a_k)}\,,
\ee
\be
\label{V_n}
\hspace{13mm}{\rm V}_n^{\triple{ijk}}({\bm a}|\bm x) =
X_{ij}^{|\bm a_{i,j}|'}
X_{jk}^{|\bm a_{j,k}|'}
X_{ki}^{|\bm a_{k,i}|'}
\prod_{b\in V_{\triangle}\vphantom{\cham_n^{(b)}}}\,
\prod_{l\in\cham_n^{(b)}}
X_{bl}^{-a_l}\,.
\ee
\end{definition}

\begin{rem}
\label{rem:legs}
The leg-factor has  two crucial properties:
\begin{enumerate}

\item Since the parameter vector $\bm a$ satisfies the constraint (\ref{conf}), the leg-factor is covariant under  the $O(D+1,1)$ conformal transformations of the point set  $\bm x$ \cite{Alkalaev:2025fgn,Alkalaev:2025zhg}.

\item  The leg-factor depends only on the squared distances $X_{rs}\in\bm X_n$ associated with all basis chords $(r,s)\in\bm C_n^{\triple{ijk}}$.

\end{enumerate}

\end{rem}

The chord decomposition $\bm C_n  = \bm C_n^{\triple{ijk}} \cup \, \overline{\bm C}_n^{\triple{ijk}}$ underlies  the construction of the conformally covariant and invariant components of the basis function, respectively. Indeed, according to Proposition \bref{rem:cross_ratio_non_basis_chord}, the arguments $\bm \rY_n^{\triple{ijk}}$ of the polygonal function  are conformally invariant cross-ratios that are  in one-to-one correspondence with the non-basis chords $\overline{\bm C}_n^{\triple{ijk}}$. On the other hand, according to Remark \bref{rem:legs}, the conformally covariant leg-factor is associated with the entire set of  basis chords. The triangle-factor serves as a normalization constant.\footnote{When the basis function is considered as a part of the decomposition of the conformal integral, this normalization constant ensures the invariance under the action of the cyclic group \cite{Alkalaev:2025zhg}.}

\section{GKZ hypergeometric systems}
\label{sec:GKZ}

\subsection{Basic ingredients}
\label{sec:basics}

In what follows, we provide a self-contained definition of the  GKZ hypergeometric system and  outline the construction of its formal solutions (for reviews see e.g.  \cite{Weinzierl:2022eaz, Cattani, Beukers, klausen2023hypergeometricfeynmanintegrals}).
\begin{rem}
Throughout this subsection, $n$, $N$, and $\xr$, which specify the dimensions of the relevant matrices, are treated as general integers subject only to $\xr=N-n$. Specialization to polygonal functions imposes further relations, as shown in the next subsection.
\end{rem}
\begin{definition}
\label{def:par_tor}
Let $\bm \alpha = (\alpha_1,...,\alpha_n) \in \mathbb{C}^n$ be a parameter vector, and let
$\cA\in\operatorname{Mat}_{n\times N}(\mathbb{Z})$ be a toric matrix, which  can be viewed as the vector configuration
\be
\cA=
\left\{
{\bm \cA}^{(1)},...,{\bm \cA}^{(N)}
\;\middle|\;
{\bm \cA}^{(p)}
=\big(\cA_1^{(p)},...,\cA_n^{(p)}\big)^\rT
\in\mathbb{Z}^n\,,
\quad p=1,...,N
\right\}\,.
\ee
The toric matrix is required to satisfy the conditions of maximal rank and homogeneity:
\begin{enumerate}

\item The column vectors in $\cA$ span $\mathbb{Z}^n$, so that $\cA$ has rank $n$.

\item There exists a linear form $h:\mathbb{Z}^n\rightarrow\mathbb{Z}$ such that $h({\bm \cA}^{(p)})=1$ for all $p=1,...,N$.

\end{enumerate}
\end{definition}
\noindent For $\cA$ to have maximal rank, it is necessary that $N\geq n$. The homogeneity condition states that all column vectors in $\cA$ lie on a common hyperplane that does not pass through the origin in $\mathbb{Z}^n$.
\begin{definition}
The lattice of (integer linear) relations among the columns of $\cA$:
\be
\label{lat}
\mathbb{L}
= {\rm ker}_{\mathbb{Z}}\cA
=\bigg\{
\bm\ell=(\ell_1,...,\ell_N)\in\mathbb{Z}^N
\;\big|\;
\sum_{p=1}^{N}\ell_p{\bm \cA}^{(p)}=0
\bigg\}\,.
\ee
\end{definition}
\noindent Since $\cA$ has rank $n$, the lattice $\mathbb{L}$ has rank $N-n\equiv\xr$. In addition, the homogeneity condition implies that its elements are constrained by
\be
\label{linear}
\sum_{p=1}^{N}\ell_p=0\,.
\ee
\begin{definition}
\label{prop:GKZ}
The GKZ hypergeometric system defined by $\cA$ and $\bm\alpha$
consists of the following differential equations for a function
$\Phi$ of $N$ variables $\bm v=(v_1,...,v_N)\in\mathbb{C}^N$:

\begin{itemize}

\item the toric equations, one for each lattice vector
$\bm\ell=(\ell_1,...,\ell_N)\in\mathbb{L}$,
\be
\label{toric_eq}
\Bigg[
\prod_{\ell_p>0}\left(\frac{\partial}{\partial v_p}\right)^{\ell_p}
-\prod_{\ell_p<0}\left(\frac{\partial}{\partial v_p}\right)^{-\ell_p}
\Bigg]\Phi(\bm v)=0\,;
\ee

\item the Euler equations
\be
\label{Euler_eq}
\Bigg[
\sum_{p=1}^{N}\cA_l^{(p)}\theta_p-\alpha_l
\Bigg]\Phi(\bm v)=0\,,
\qquad
\theta_p=v_p\frac{\partial}{\partial v_p}\,,
\qquad
l=1,...,n\,.
\ee

\end{itemize}
\end{definition}
\begin{prop}
\label{prop:solution}
The following $\Gamma$-series
\be
\label{Gamma_series}
\Phi_{\mathbb{L},\bm \gamma}(\bm v)
=\sum_{\bm\ell\in\mathbb{L}}\,
\prod_{p=1}^{N}\,
\frac{v_p^{\ell_p+\gamma_p}}
{\Gamma(\ell_p+\gamma_p+1)}
\ee
is a formal solution to the GKZ system (\ref{toric_eq})--(\ref{Euler_eq}), provided that the parameter vector $\bm\gamma=(\gamma_1,...,\gamma_N)\in\mathbb{C}^N$ satisfies the linear equations
\be
\label{alpha_A_gamma}
\bm\alpha=\cA\bm\gamma\,.
\ee
\end{prop}
\begin{definition}
The Gale dual of $\cA$ is the matrix
$\cG\in\operatorname{Mat}_{N\times\xr}(\mathbb{Z})$, whose columns form a basis of the lattice of relations $\mathbb{L}$. Equivalently, $\cG$ can be viewed as the vector configuration
\be
\label{Gale_dual_first}
\cG
=\left\{
{\bm \cG}^{(1)},...,{\bm \cG}^{(\xr)}
\;\middle|\;
{\bm \cG}^{(p)}
=\big(\cG_1^{(p)},...,\cG_N^{(p)}\big)^\rT
\in\mathbb{Z}^N\,,
\quad p=1,...,\xr
\right\}\,.
\ee
\end{definition}

 Formally, the toric part of the GKZ system consists of infinitely many equations. The Gale dual serves to organize this infinite family. Then, every $\bm\ell\in\mathbb{L}$ admits a unique representation
\be
\label{lattice_vector_expansion}
\bm\ell
=\sum_{p=1}^{\xr}{\bm \cG}^{(p)}m_p
\,,
\qquad
{\bm \cG}^{(p)}=(\cG_1^{(p)},...,\cG_N^{(p)})^\rT\in\mathbb{Z}^N \,,
\qquad
{\bm m}=(m_1,...,m_\xr)\in\mathbb{Z}^\xr\,.
\ee
\begin{rem}
Equivalently, the Gale dual of $\cA$ is any matrix $\cG$ satisfying
\be
\cA\cG=0\,,
\qquad
{\rm rank}\,\cG=\xr\,,
\qquad
{\rm im}_{\mathbb{Z}}\,\cG={\rm ker}_{\mathbb{Z}}\,\cA\,.
\ee
\end{rem}
\begin{rem}
\label{rem:freedom}
The matrices representing the GKZ data are not unique. The transformation $(\cA,\bm\alpha)\longmapsto(U\cA,U\bm\alpha)$, where $U\in{\rm GL}(n,\mathbb{Z})$,
leaves the lattice $\mathbb{L}$ intact and reorganizes the Euler equations as particular linear combinations of the original ones. A permutation of the columns of $\cA$ amounts to a relabeling of the GKZ variables $\bm v$. Independently, the transformation $\cG\longmapsto\cG T$, where
$T\in{\rm GL}(\xr,\mathbb{Z})$,  corresponds to changing the basis in $\mathbb{L}$.
\end{rem}

The $\Gamma$-series \eqref{Gamma_series} can be rewritten directly in terms of the Gale dual:
\be
\label{eq:gamma-series-gale-dual}
\Phi_{\mathbb{L},\bm \gamma}({\bm v})
\equiv\Phi_{\cG,\bm \gamma}({\bm v})
=\sum_{{\bm m}\in\mathbb Z^\xr}\,
\prod_{s=1}^{N}\,
\frac{v_s^{\gamma_s+(\cG{\bm m})_s}}
{\Gamma\!\left(\gamma_s+(\cG{\bm m})_s+1\right)}\,.
\ee
Here, $(\cG{\bm m})_s$ denotes the $s$-th component of the vector $\cG{\bm m}$. For the GKZ systems considered below, the toric matrix can  always be brought  to the following block  form (see Remark \bref{rem:freedom}):
\be
\label{matrix_A_another_link}
\cA=\left(\mathbb{1}_n\mid\cB\right)
\in\operatorname{Mat}_{n\times (n+\xr)}(\mathbb{Z})\,,
\qquad
\text{where} \quad  \cB\in\operatorname{Mat}_{n\times\xr}(\mathbb{Z})\,.
\ee
Then, the Gale dual and the solution of \eqref{alpha_A_gamma} can be chosen as
\be
\label{normalized_Gale_dual}
\cG =
\begin{pmatrix}
-\cB \\[2pt]
\phantom{-}\mathbb{1}_\xr
\end{pmatrix}
\qquad \text{and}\qquad
{\bm \gamma}=
\begin{pmatrix}
{\bm \alpha}\\
{\bm {0}}_\xr
\end{pmatrix},
\ee
where $\mathbb{1}_\xr$ is the $\xr\times\xr$ identity matrix and  ${\bm {0}}_\xr$ is the $\xr$-component zero vector.
With these choices, the $\Gamma$-series \eqref{eq:gamma-series-gale-dual} takes the form
\be
\label{eq:gamma-series-block-form}
\Phi_{\cG,\bm \gamma}({\bm v})
\equiv
\Phi_{{\cB},\bm \alpha}({\bm v})
=V_{\bm \alpha}(\bm v)\, \sum_{\bm m \in \ZZ^\xr_{+}}
\left[
\prod_{l=1}^{n}
\frac{v_l^{-(\cB{\bm m})_l}}
{\Gamma\!\left(\alpha_l-(\cB{\bm m})_l+1\right)}
\right]
\prod_{p=1}^{\xr}\frac{v_{n+p}^{m_p}}{m_p!}\,,
\ee
where the factor independent of the summation variables $\bm m$ is explicitly factored out:
\be
\label{GKZ_leg}
V_{\bm \alpha}(\bm v) = \prod_{l=1}^{n} v_l^{\alpha_l} \,.
\ee
Note that the summation domain is restricted to $\ZZ^\xr_+$, since $1/\Gamma(m_p+1)=0$ for negative integers  $m_p<0$.

The role of  $V_{\bm\alpha}(\bm v)$ is to ensure that the $\Gamma$-series \eqref{eq:gamma-series-block-form} satisfies the Euler equations \eqref{Euler_eq}. For the toric matrix \eqref{matrix_A_another_link}, these equations require
\be
\bigg[\theta_l+\sum_{p=1}^{\xr}\cB_l^{(p)}\theta_{n+p}\bigg] \Phi_{{\cB},\bm \alpha}({\bm v}) = \alpha_l\, \Phi_{{\cB},\bm \alpha}({\bm v})\,,
\qquad
l\in \NN_n\,.
\ee
Applying the operator on the left-hand side to the sum in \eqref{eq:gamma-series-block-form} gives zero, since in each term the contribution from $\theta_l$ cancels the combined contribution from the remaining derivatives. On the other hand, acting on $V_{\bm\alpha}(\bm v)$ yields  $\alpha_l V_{\bm\alpha}(\bm v)$, because it contains $v_l^{\alpha_l}$ and is independent of all $v_{n+p}$. The Leibniz product rule  yields the Euler equations for the full $\Gamma$-series.

Thus, for the toric matrix \eqref{matrix_A_another_link}, the formal solution to the corresponding GKZ system is fully determined by the matrix $\cB\in\operatorname{Mat}_{n\times\xr}(\mathbb{Z})$ and the parameter vector $\bm \alpha$. In what follows, we work entirely within this representation.

\subsection{GKZ data for the polygonal function}
\label{sec:GKZ_data}
In this section, we reconstruct the GKZ data associated with the $n$-point polygonal function directly from its series representation \eqref{general_polygonal}. We first recast the coefficients of the polygonal power series in a form that can be directly compared with the coefficients of the GKZ solution \eqref{eq:gamma-series-block-form}, and thereby extract the corresponding GKZ data. We then substitute these data into \eqref{eq:gamma-series-block-form} and show that the resulting $\Gamma$-series coincides with the basis function \eqref{basisfunc} up to an overall coordinate-independent factor.

Using the Euler reflection formula, the  polygonal function \eqref{general_polygonal} can be represented as
\be
\label{H_in_the_jungle}
\ba{l}
\dps
\HG^{\triple{ijk}}_n({\bm a}|\bm \rY^{\triple{ijk}}_n)
=
\dps \frac{{\rm F}_{n}^{\triple{ijk}}({\bm a})}{{\rm S}_{n}^{\triple{ijk}}({\bm a})}
\sum_{\bm m \in \mathbb{Z}_{+}^{\xr}}\,
\prod_{p=1}^{\xr}\,
\frac{
\big(\rY^{\triple{ijk}}_p\big)^{m_p}}{m_p!}
\prod_{l\notin\{i,j,k\}}
\frac{1}{\Gamma(1-a_l-M_l)}
\vspace{2mm}
\\
\dps
\hspace{25mm}\times\;
\frac{1}{
\Gamma(1+|{\bm a}_{i,j}|'+M_{ij})
\Gamma(1+|{\bm a}_{j,k}|'+M_{jk})
\Gamma(1+|{\bm a}_{k,i}|'+M_{ki})}\,,
\ea
\ee
where the cyclic sums $|{\bm a}_{r,s}|'$ and the triangle-factor ${ \rm S}_{n}^{\triple{ijk}}({\bm a})$ are defined in \eqref{aa} and \eqref{star_n}, respectively, the prefactor ${\rm F}_{n}^{\triple{ijk}}({\bm a})$ is given by
\be
\label{pref_ya}
\ba{l}
\dps
{\rm F}_{n}^{\triple{ijk}}({\bm a})
=
-\frac{\pi^n}
{ \dps\prod_{l=1}^{n}\Gamma(a_l)
 \dps\prod_{l\notin\{i,j,k\}}\sin(\pi a_l)}
\frac{1}
{\sin\!\big(\pi|{\bm a}_{i,j}|'\big)
 \sin\!\big(\pi|{\bm a}_{j,k}|'\big)
 \sin\!\big(\pi|{\bm a}_{k,i}|'\big)}\,.
\ea
\ee
The combinations of the summation variables  $M_l$ and $M_{st}$ \eqref{MM} can be written in terms of the transianic matrix \eqref{tr_mat}:
\be
\label{M_Bm}
\ba{c}
M_l=(\cB_n^{\triple{ijk}}\bm m)_l\,,\qquad l\notin \{i,j,k\}\,,
\\
M_{ij}=-(\cB_n^{\triple{ijk}}\bm m)_i\,,
\quad
M_{jk}=-(\cB_n^{\triple{ijk}}\bm m)_j\,,
\quad
M_{ki}=-(\cB_n^{\triple{ijk}}\bm m)_k\,.
\ea
\ee
Additionally, substituting the representation of the cross-ratio $\rY_p^{\triple{ijk}}$ \eqref{diagrammatic_cross_ratio} into \eqref{H_in_the_jungle}, we obtain
\be
\label{polygonal_GKZ_form}
\ba{c}
\dps
\HG^{\triple{ijk}}_n({\bm a}|\bm \rY^{\triple{ijk}}_n)
=
\dps \frac{{\rm F}_{n}^{\triple{ijk}}({\bm a})}{{\rm S}_{n}^{\triple{ijk}}({\bm a})}
\sum_{\bm m \in \mathbb{Z}_{+}^{\xr}}\,
\prod_{p=1}^{\xr}\frac{X_{r_ps_p}^{m_p}}{m_p!}
\vspace{2mm}
\bigg[
\prod_{a\in V_{\triangle}\vphantom{\cham_n^{(a)}}}\,
\prod_{l\in\cham_n^{(a)}}
\frac{X_{al}^{-(\cB_n^{\triple{ijk}}\bm m)_l}}
{\Gamma\!\left(1-a_l-(\cB_n^{\triple{ijk}}\bm m)_l\right)}
\bigg]
\\
\dps
\times\,
\frac{
X_{ij}^{-(\cB_n^{\triple{ijk}}\bm m)_i}
X_{jk}^{-(\cB_n^{\triple{ijk}}\bm m)_j}
X_{ki}^{-(\cB_n^{\triple{ijk}}\bm m)_k}}
{\Gamma\!\left(1+|{\bm a}_{i,j}|'-(\cB_n^{\triple{ijk}}\bm m)_i\right)
 \Gamma\!\left(1+|{\bm a}_{j,k}|'-(\cB_n^{\triple{ijk}}\bm m)_j\right)
 \Gamma\!\left(1+|{\bm a}_{k,i}|'-(\cB_n^{\triple{ijk}}\bm m)_k\right)}\,.
\ea
\ee

Formula \eqref{polygonal_GKZ_form} shares the functional structure of the GKZ solution \eqref{eq:gamma-series-block-form} with a block-form toric matrix \eqref{matrix_A_another_link}. Based on  this comparison, the GKZ data for the $n$-point polygonal function can be directly extracted and are summarized below.
\begin{definition}
\label{def:GKZ_polyg}
The upper block of the Gale dual matrix (\ref{Gale_dual_first}), (\ref{normalized_Gale_dual}) is given by the negative of the transianic matrix (\ref{tr_mat}):
\be
\label{Gale_mat_polygonal}
\cG_n^{\triple{ijk}}=
\begin{pmatrix}
-\cB_n^{\triple{ijk}}\\[2pt]
\phantom{-}\mathbb{1}_\xr
\end{pmatrix}
\in\operatorname{Mat}_{N\times\xr}(\mathbb{Z})\,,
\qquad
N=\frac{n(n-1)}{2}\,.
\ee
The corresponding toric matrix takes the form
\be
\label{toric_pol}
\cA_n^{\triple{ijk}} = \Big(\mathbb{1}_n \mid \cB_n^{\triple{ijk}}\Big)
\in\operatorname{Mat}_{n\times N}(\mathbb{Z})\,.
\ee
The parameter vectors are given by
\be
\label{alpha_polyg}
\bm \gamma^{\triple{ijk}}
=\big(\bm \alpha^{\triple{ijk}}, \mathbb{0}_{\xr}\big)
\in\mathbb{R}^{N}\,,\;
\;\;
\bm \alpha^{\triple{ijk}}\in\mathbb{R}^n:
\;\;
\big(\bm \alpha^{\triple{ijk}}\big)_l=
\begin{cases}
\;\;\;\;-a_l\,, & l\notin \{i,j,k\}\,,\\[2mm]
\phantom{-}|{\bm a}_{i,j}|'\,, & l=i\,,\\[2mm]
\phantom{-}|{\bm a}_{j,k}|'\,, & l=j\,,\\[2mm]
\phantom{-}|{\bm a}_{k,i}|'\,, & l=k\,.
\end{cases}
\ee
\end{definition}
\noindent By construction, both parameter vectors  satisfy the general relation \eqref{alpha_A_gamma}.
Also, note that the ordering of the rows of the Gale dual matrix \eqref{Gale_mat_polygonal} can be different, which   reflects the chosen ordering of the Gamma functions in the denominator of \eqref{polygonal_GKZ_form}. Here, we order these factors so that the second transianic indices $b_{ij}^{(p)}$, $b_{jk}^{(p)}$,  $b_{ki}^{(p)}$ occupy the $i$-th, $j$-th, $k$-th rows, respectively. Any other ordering is equally admissible provided that the corresponding components of $\boldsymbol{\gamma}$ and the associated variables are permuted consistently.
\begin{definition}
\label{def:gale_lat_poly}
A basis of the lattice of relations $\mathbb{L}$ is given by the vectors (cf. (\ref{Gale_dual_first})):
\be
\label{Gale_vector_configuration_polygonal}
{\bm \cG}^{(p)}
=
\begin{pmatrix}
-{\bm \cB}^{(p)}\\[1mm]
\phantom{-}{\bm e}^{(p)}
\end{pmatrix}
\in\mathbb{Z}^{N}\,,
\quad
p=1,...,\xr\,.
\ee
Here, ${\bm e}^{(p)}=(0,...,1,...,0)\in\mathbb{Z}^{\xr}$ is the $p$-th canonical basis vector, and ${\bm \cB}^{(p)}$ is read off from the transianic matrix (cf. (\ref{transianic_vector_configuration})).
\end{definition}
\noindent It follows that there is a one-to-one correspondence between $\rY_p^{\triple{ijk}}$ and ${\bm \cG}^{(p)}$ for all $p=1,...,\xr$. In other words, the cross-ratio set constructed by the diagrammatic algorithm uniquely defines a basis of the lattice of relations.

Next, comparing the arguments of the GKZ solution \eqref{eq:gamma-series-block-form} with the arguments of the polygonal function
\eqref{polygonal_GKZ_form}, we establish  the following identification of variables.
\begin{prop}
\label{prop:mapping_new}
Given a basis triangle $\triangle_n^{\triple{ijk}}\subset P_n$ and the respective chord decomposition (\ref{chord_set_decomposition}),  the GKZ variables $\bm v=(v_1,...,v_N)$ are  identified with the squared distances $\bm X_n$ as follows:
\begin{itemize}

\item for the basis chords,
\be
\label{cross_ratio_depend_on_v0_new}
\ba{c}
v_l=X_{al}\,,\quad  l\in\cham_n^{(a)}\,, \; a\in V_{\triangle}\,,\\
v_i=X_{ij}\,, \quad v_j=X_{jk}\,, \quad v_k=X_{ki}\,;
\ea\ee

\item for the non-basis chords,
\be
\label{cross_ratio_depend_on_v_new}
v_{n+p}=X_{r_ps_p}\,,
\qquad
p=1,...,\xr\,,
\ee
where $(r_p,s_p)\in\overline{\bm C}_n^{\triple{ijk}}$ is the unique non-basis chord associated with the cross-ratio $\rY_p^{\triple{ijk}}$, as stated in Proposition  \bref{rem:cross_ratio_non_basis_chord}.

\end{itemize}
\end{prop}
\noindent  Following Remark \bref{rem:diagrammatic_interpretation}, we emphasize the role of the chord decomposition \eqref{chord_set_decomposition}. For a fixed basis triangle, this decomposition naturally partitions the GKZ variables: the first $n$ variables are indexed by the basis chords, while the remaining $\xr$ variables correspond to the non-basis ones. By construction, the latter are in a one-to-one correspondence with the cross-ratios generated by the diagrammatic algorithm.
\begin{prop}
\label{prop:Gale_cross_ratios}
Under the variable identification introduced in Proposition \bref{prop:mapping_new}, the cross-ratios \eqref{diagrammatic_cross_ratio} take the form
\be
\label{cross_ratio_Gale_monomial}
\rY_p^{\triple{ijk}}
=
\prod_{s=1}^{N}v_s^{\cG_s^{(p)}}
=
\prod_{(a,b)\in\bm C_n}X_{ab}^{\cG_{(a,b)}^{(p)}}\,,
\qquad
p=1,...,\xr\,,
\ee
where the rows of the Gale dual are labelled by chords $(a,b)\in\bm C_n$, and its entries are given by
\be
\label{Gale_matrix_chord_entries}
\cG_{(a,b)}^{(p)}
=
\begin{cases}
\phantom{-}1\,,
& (a,b)\in E_{\textcolor{green}{\bm\times}}^{(p)}\,,
\\
-1\,,
& (a,b)\in E_{\textcolor{orange}{\bm\times}}^{(p)}\,,
\\
\phantom{-}0\,,
& \text{otherwise}\,.
\end{cases}
\ee
Thus, the $p$-th column of the Gale dual collects the powers of all squared distances entering the corresponding cross-ratio.
\end{prop}
\begin{proof}
A direct consequence of Proposition \bref{prop:diagrammatic_cross_ratio} and the variable identification of Proposition \bref{prop:mapping_new} is
\be
\label{cross_ratio_GKZ_variables}
\rY_p^{\triple{ijk}}
=
v_{n+p}\prod_{l=1}^{n}v_l^{-\cB_l^{(p)}}\,,
\qquad
p=1,...,\xr\,,
\ee
where $\cB_l^{(p)}$ are the elements of the transianic matrix \eqref{tr_mat}. Thus, each cross-ratio contains a single GKZ variable $v_{n+p}$ associated with its corresponding non-basis chord, while the remaining factor involves only the variables associated with basis chords. According to \eqref{Gale_vector_configuration_polygonal}, the first $n$ components of the $p$-th column ${\bm\cG}^{(p)}$ of the Gale dual are given by the negative transianic vector $-{\bm\cB}^{(p)}$, while its remaining $\xr$ components form ${\bm e}^{(p)}$. Hence, the powers of the GKZ variables in \eqref{cross_ratio_GKZ_variables} are given by ${\bm\cG}^{(p)}$, which proves the first equality in \eqref{cross_ratio_Gale_monomial}. The second  equality follows from Proposition \bref{prop:mapping_new}. Finally, green and orange chords contribute squared distances to the numerator and denominator, respectively, whereas absent chords contribute no factors. This yields \eqref{Gale_matrix_chord_entries}.
\end{proof}
\begin{rem}
\label{rem:Gale_adjacency}
The Gale dual matrix in the form (\ref{Gale_matrix_chord_entries}) coincides with the  signed adjacency matrix of the cross-ratio diagram (\ref{adjacency_cross_ratio_matrix}), see Appendix \bref{app:adjacency} for more details.
\end{rem}
\begin{rem}
The parameter vector (\ref{alpha_polyg}) satisfies the constraint
\be
\label{hypersurf}
\sum_{l=1}^n \big(\bm \alpha^{\triple{ijk}}\big)_l=-\frac{D}{2}\,.
\ee
Thus, the parameter vectors $\bm a$ and $\bm \alpha^{\triple{ijk}}$ both belong to $\mathbb{R}^{n}$ but lie on different hyperplanes, cf. (\ref{conf}) and (\ref{hypersurf}).
\end{rem}
\begin{rem}
\label{rem:GKZ_leg}
Under the variable identification introduced in  Proposition
\bref{prop:mapping_new}, the factor (\ref{GKZ_leg}) coincides with the leg-factor produced by the diagrammatic algorithm:
\be
V_{\bm \alpha^{\triple{ijk}}}(\bm v)
=
{\rm V}_n^{\triple{ijk}}({\bm a}|\bm x)\,,
\ee
where $\bm \alpha^{\triple{ijk}}$ is constrained by \eqref{hypersurf}.
This identity ensures the consistency of the diagrammatic algorithm with the GKZ approach at this stage of the construction.
\end{rem}

Summarizing these results, we establish the connection between the polygonal function and  the  $\Gamma$-series solution  of the GKZ hypergeometric system.
\begin{theorem}
The $\Gamma$-series representation for the $n$-point polygonal function $\HG^{\triple{ijk}}_n({\bm a}|\bm \rY^{\triple{ijk}}_n)$ (\ref{general_polygonal}) has the form
\be
\label{Gamma_series_with_F_V_H}
\ba{l}
\dps
\Phi_{\cB_n^{\triple{ijk}},\,\bm\alpha^{\triple{ijk}}}(\bm v)
=
\frac{1}{{\rm F}_{n}^{\triple{ijk}}({\bm a})} \,
{\rm S}_n^{\triple{ijk}}({\bm a})
{\rm V}_n^{\triple{ijk}}({\bm a}|\bm x)\,
\HG_n^{\triple{ijk}}(\bm a|\bm\rY_n^{\triple{ijk}})\,.
\ea
\ee
Here, ${ \rm S}_{n}^{\triple{ijk}}({\bm a})$, ${\rm V}_n^{\triple{ijk}}({\bm a}|\bm x)$, and ${\rm F}_{n}^{\triple{ijk}}({\bm a})$ are given by (\ref{star_n}), (\ref{V_n}), and (\ref{pref_ya}), respectively. The GKZ data expressed in  diagrammatic terms are gathered in Definition \bref{def:GKZ_polyg}.  This $\Gamma$-series is identified, up to a constant prefactor, with the basis function:

\be
\label{Gamma_series_Phi}
\Phi_{\cB_n^{\triple{ijk}},\, \bm \alpha^{\triple{ijk}}}(\bm v) = \Phi_{\cB_n^{\triple{ijk}},\bm \alpha^{\triple{ijk}}}(\bm X_n)
\;\propto\;
\Phi_n^{\triple{ijk}}(\bm a|\bm x)\,,
\ee
which is given in Definition \bref{def:bas}. The variable identification  is given in Proposition \bref{prop:mapping_new}.
\end{theorem}

\subsection{GKZ hypergeometric  system for the polygonal function}
\label{sec:GKZ_polygon}

\begin{theorem}[toric equations]
\label{prop:general_toric_equation}
For every $\bm m=(m_1,...,m_\xr)\in\mathbb{Z}^{\xr}$, the $\Gamma$-series associated with the $n$-point polygonal function \eqref{Gamma_series_with_F_V_H} satisfies the toric equation
\be
\label{general_toric_equation_compact}
\Bigg[
\prod_{\substack{(a,b)\in\bm C_n\\ \ell_{ab}>0}}
\left(\frac{\partial}{\partial X_{ab}}\right)^{\ell_{ab}}
-
\prod_{\substack{(a,b)\in\bm C_n\\ \ell_{ab}<0}}
\left(\frac{\partial}{\partial X_{ab}}\right)^{-\ell_{ab}}
\Bigg]
\Phi_{\cB_n^{\triple{ijk}},\bm\alpha^{\triple{ijk}}}(\bm X_n)=0\,,
\ee
where $\ell_{ab}$ denotes the component of the lattice vector $\bm\ell=\sum_{p=1}^{\xr}m_p{\bm\cG}^{(p)}\in\mathbb{L}$ associated with the chord $(a,b)\in\bm C_n$ and is given by
\be
\label{general_lattice_vector_component_diagrammatic}
\ell_{ab}
=\smashoperator{\sum_{\substack{p=1\\
(a,b)\in E^{(p)}_{\textcolor{green}{\bm\times}}}}^{\xr}}m_p
\quad-\quad
\smashoperator{\sum_{\substack{p=1\\
(a,b)\in E^{(p)}_{\textcolor{orange}{\bm\times}}}}^{\xr}}m_p\,.
\ee
\end{theorem}
\begin{proof}
Taking the component associated with $(a,b)$ in \eqref{lattice_vector_expansion} and using \eqref{Gale_matrix_chord_entries}, we obtain $\ell_{ab}=\sum_{p=1}^{\xr}m_p\cG_{(a,b)}^{(p)}$, which gives \eqref{general_lattice_vector_component_diagrammatic}. Under the variable identification of Proposition \bref{prop:mapping_new}, substituting these components into \eqref{toric_eq} yields \eqref{general_toric_equation_compact}. Since the columns of the Gale dual  form a basis of $\mathbb{L}$, every lattice vector is obtained as $\bm m$ ranges over $\mathbb{Z}^{\xr}$.
\end{proof}
\begin{theorem}[Euler equations]
\label{prop:diagrammatic_GKZ}
\label{proposition_for_Euler_eq}
The $\Gamma$-series associated with the $n$-point polygonal function (\ref{Gamma_series_with_F_V_H}) satisfies the following $n$ Euler equations, one for each basis chord in $\bm C_n^{\triple{ijk}}$.

\begin{itemize}
\item For any $(a,b)\in E_{\triangle} \subset \bm C_n^{\triple{ijk}}$:
\end{itemize}
\be
\label{Euler_triangle_diagrammatic}
\Bigg[
X_{ab}\frac{\partial}{\partial X_{ab}}
+
\smashoperator{\sum_{\substack{p=1\\(a,b)\in E^{(p)}_{\textcolor{orange}{\bm\times}}}}^{\xr}}
X_{r_p s_p}\frac{\partial}{\partial X_{r_p s_p}}
-
\smashoperator{\sum_{\substack{p=1\\(a,b)\in E^{(p)}_{\textcolor{green}{\bm\times}}}}^{\xr}}
X_{r_p s_p}\frac{\partial}{\partial X_{r_p s_p}}
-|{\bm a}_{a,b}|'
\Bigg]
\Phi_{\cB_n^{\triple{ijk}},\,\bm \alpha^{\triple{ijk}}}(\bm X_n)=0\,.
\ee

\begin{itemize}
\item For any $(a,b)\in\bm C_n^{\triple{ijk}}\setminus E_{\triangle}$ with $b\in\cham_n^{(a)}$:
\end{itemize}
\be
\label{Euler_basis_diagrammatic}
\Bigg[
X_{ab}\frac{\partial}{\partial X_{ab}}
+
\smashoperator{\sum_{\substack{p=1\\(a,b)\in E^{(p)}_{\textcolor{orange}{\bm\times}}}}^{\xr}}
X_{r_p s_p}\frac{\partial}{\partial X_{r_p s_p}}
+a_b
\Bigg]
\Phi_{\cB_n^{\triple{ijk}},\,\bm \alpha^{\triple{ijk}}}(\bm X_n)=0\,.
\ee
In equations (\ref{Euler_triangle_diagrammatic}) and (\ref{Euler_basis_diagrammatic}), $(r_p,s_p)\in \overline{\bm C}_n^{\triple{ijk}}$ is the unique non-basis chord entering the cross-ratio diagram of  $\rY_p^{\triple{ijk}}\in \bm \rY_n^{\triple{ijk}}$, as specified in Proposition \bref{rem:cross_ratio_non_basis_chord}.
\end{theorem}
\begin{proof}
Substituting the toric  matrix \eqref{toric_pol} into the Euler equations \eqref{Euler_eq} yields
\be
\label{Euler_original_diagrammatic}
\Bigg[
\theta_l
+\sum_{p=1}^{\xr}\cB_l^{(p)}\theta_{n+p}
-\big(\bm \alpha^{\triple{ijk}}\big)_l
\Bigg]
\Phi_{\cB_n^{\triple{ijk}},\bm \alpha^{\triple{ijk}}}(\bm v)=0\,,
\qquad l=1,...,n\,.
\ee
Here, $\theta_s=v_s\partial/\partial v_s$. According to  Proposition \bref{prop:mapping_new}, $\theta_{n+p}=X_{r_p s_p}\partial/\partial X_{r_p s_p}$, while the first $n$ variables correspond to the basis chords \eqref{set_of_basic_chords}.

For the row of \eqref{Gale_matrix_chord_entries} corresponding to the chord $(i,j)$, we have $v_i=X_{ij}$, hence $\theta_i=v_i \partial/\partial v_i = X_{ij}\partial/\partial X_{ij}$, and $(\bm\alpha^{\triple{ijk}})_i=|{\bm a}_{i,j}|'$ by \eqref{alpha_polyg}. Since $\cB_i^{(p)}=-\cG_i^{(p)}$, the relation \eqref{Gale_matrix_chord_entries} makes the sum in \eqref{Euler_original_diagrammatic} equal to the orange-chord contribution minus the green-chord contribution. For the chord $(a,b)=(i,j)$, this yields  \eqref{Euler_triangle_diagrammatic}. The rows associated with the chords $(j,k)$ and $(k,i)$ follow by cyclic permutation. The ordered representatives $(i,j),(j,k),(k,i)$ also specify the cyclic parameter sums $|{\bm a}_{a,b}|'$ in these three equations.

For the remaining basis chords $(a,l)$ with $l\in\cham_n^{(a)}$, the variable and parameter identifications are $v_l=X_{al}$ and $(\bm\alpha^{\triple{ijk}})_l=-a_l$, respectively. By \eqref{diagrammatic_cross_ratio}, $(a,l)$ can occur only in the denominator of a cross-ratio. Thus, \eqref{Gale_matrix_chord_entries} and $\cB_l^{(p)}=-\cG_l^{(p)}$ select only the orange-chord terms in the sum in \eqref{Euler_original_diagrammatic}. The substitution of $\theta_l= v_l \partial/\partial v_l = X_{al}\partial/\partial X_{al}$, then yields \eqref{Euler_basis_diagrammatic}.
\end{proof}

Choosing  $\bm m={\bm e}^{(p)}$ in Theorem \bref{prop:general_toric_equation} yields the lattice vector  $\bm\ell={\bm\cG}^{(p)}$. As $p$ ranges from $1$ to $\xr$, these choices select the $\xr$ basis vectors of $\mathbb{L}$ given by the columns of the Gale dual. {The relation (\ref{Gale_matrix_chord_entries}) then shows that the positive and negative parts select exactly the green and orange chords, respectively, with all nonzero exponents equal to $1$.}
\begin{cor}[a finite toric subsystem]
\label{cor:finite_toric_subsystem}
Choosing the basis vectors of the lattice of relations  reduces
(\ref{general_toric_equation_compact}) to $n(n-3)/2$ toric equations of the form
\be
\label{toric_proposition}
\Bigg[
\prod_{(a,b)\in E^{(p)}_{\textcolor{green}{\bm\times}}}
\frac{\partial}{\partial X_{ab}}
-
\prod_{(c,d)\in E^{(p)}_{\textcolor{orange}{\bm\times}}}
\frac{\partial}{\partial X_{cd}}
\Bigg]
\Phi_{\cB_n^{\triple{ijk}},\,\bm \alpha^{\triple{ijk}}}
(\bm X_n)=0\,.
\ee
\end{cor}
\begin{rem}
\label{rem:diagrammatic_subsystem}
The complete GKZ hypergeometric system is obtained by combining the toric equations from  Theorem \bref{prop:general_toric_equation} with the Euler equations  from Theorem \bref{proposition_for_Euler_eq}. Corollary \bref{cor:finite_toric_subsystem} singles out a finite subsystem of PDEs that, although associated with a basis of the lattice
$\mathbb{L}$, does not necessarily generate all toric equations.
\end{rem}
For generic values of $\bm a$, the finite toric subsystem (\ref{toric_proposition}) is sufficient to recover the polygonal series (\ref{general_polygonal}) up to a constant factor. This corresponds to  solving the differential equations using the Frobenius method. Specifically, substituting a power series with arbitrary coefficients in the same variables into the toric equations (\ref{toric_proposition}) obtained above yields recurrence relations for these coefficients. These relations then completely determine the coefficients, thereby reproducing  the desired polygonal series.
\begin{rem}
The chord decomposition (\ref{chord_set_decomposition}) determines  the way we label equations in Corollary \bref{cor:finite_toric_subsystem} and Theorem \bref{proposition_for_Euler_eq}: the $\xr$ non-basis chords label the basis toric equations, while the $n$ basis chords label the Euler equations.
\end{rem}
\begin{rem}
\label{remark_tells_the_order}
For the  finite subsystem in Corollary \bref{cor:finite_toric_subsystem}, the differential order coincides with the order of the corresponding cross-ratio. Specifically, the quadratic cross-ratios in ${}^{_1}\hspace{-0.5mm}\bm\rU_n^{\triple{ijk}} \cup{}^{_2}\hspace{-0.5mm}\bm\rU_n^{\triple{ijk}}$ correspond to   second-order equations, while the cubic cross-ratios in $\bm\rW_n^{\triple{ijk}}$ correspond to   third-order equations. The full family (\ref{general_toric_equation_compact}) also contains equations of higher order.
\end{rem}

Appendix \bref{sec:examples} illustrates the Euler equations (Theorem \bref{proposition_for_Euler_eq}) and a finite subsystem of toric equations (Corollary \bref{cor:finite_toric_subsystem}) for several low-point polygonal functions. These include the fourth Appell function $F_4$ for $n=4$ and the Srivastava--Daoust functions for $n=5,6$.

\section{Conclusions}
\label{sec:conclusion}

In this paper, we have constructed the GKZ hypergeometric systems associated with $n$-point polygonal functions, which define multipoint one-loop parametric conformal integrals in arbitrary dimensions. Starting from their power series representation, we have shown that the transianic matrix determines the nontrivial block of the toric matrix, the corresponding Gale dual, and, therefore, a basis of the lattice of relations. The propagator powers determine the GKZ parameter vectors, while the  linear identities for the transianic indices ensure the homogeneity condition for the toric matrix and imply  that all convergence indices of the polygonal functions vanish. The latter property ensures that the polygonal functions possess non-empty convergence domains in the coordinate space.

The GKZ variables identified with the squared distances are organized according to the chord decomposition. The non-basis chords correspond one-to-one with the cross-ratios that serve as arguments of the polygonal function. Meanwhile, the leg-factor is built from the variables associated with the basis chords. Within the diagrammatic algorithm, these two components are constructed independently. They combine, up to a coordinate-independent factor, into the basis function, which thus naturally arises within the GKZ framework.

The GKZ data extracted from the polygonal function yield a diagrammatic formulation of both parts of the GKZ system. The Euler equations inherit the decomposition into the edges of the basis triangle and the remaining basis chords. For an arbitrary lattice vector, the two differential monomials in the corresponding toric equation are determined by the green and orange chords of the full cross-ratio set. Restricting to the basis vectors selected by the diagrammatic algorithm singles out a finite subsystem of $n(n-3)/2$ equations: quadratic cross-ratios give second-order equations, while cubic cross-ratios give third-order ones. Although this finite subsystem is sufficient to reconstruct the polygonal function, it does not, in general, generate the full toric ideal, and the complete toric family also contains higher-order equations.

Low-point representatives of polygonal functions include the familiar hypergeometric functions \cite{Alkalaev:2025zhg}: the fourth Appell function $F_4$ for $n=4$ and the Srivastava--Daoust functions \cite{SrivastavaDaoust,Srivastava1985MultipleGH} for $n=5,6$. In Appendix \bref{sec:examples}, we have considered the corresponding GKZ hypergeometric systems as illustrations of the general results obtained in the main text. In particular, for $n=4$, we have  reproduced the second-order differential system for the fourth Appell function $F_4$. For $n=5$ and $n=6$, we have obtained systems of $5$ and $6$ Euler equations, respectively, together with finite subsystems of $5$ and $9$ toric equations. In both cases, the toric subsystems are PDEs of  second and third order.

Several questions deserve further investigation. In particular, simple linear combinations reorganize the Euler equations \eqref{Euler_triangle_diagrammatic}--\eqref{Euler_basis_diagrammatic} into the triangle-independent form \cite{ZAM:workinprogress}:
\be
\Bigg[
\sum_{\substack{m=1\\m\neq l}}^n
X_{lm}\frac{\partial}{\partial X_{lm}}
+a_l
\Bigg]
\Phi_{\cB_n^{\triple{ijk}},\,\bm\alpha^{\triple{ijk}}}(\bm X_n)
=0\,,
\qquad
l=1,...,n\,.
\ee
Formally, this system can be viewed as a new set of Euler equations defined by the toric matrix\footnote{Note that the toric matrix $\cA_n$ can be viewed as the incidence matrix of the complete graph $K_n=(V_{\hexagon},\bm C_n)$ constructed on the vertex set of the conformal polygon.} and the parameter vector
\be
\cA_n
=
\big(\cA_{l}^{(a,b)}\big)
\in\operatorname{Mat}_{n\times N}(\mathbb{Z}) \,,
\qquad
\cA_{l}^{(a,b)}
=
\delta_{l}^a+\delta_{l}^b \,,
\qquad
\bm\alpha=-\bm a\,,
\ee
where the columns are indexed by the chords $(a,b) \in \bm C_n$. The same toric matrix was previously obtained for the full conformal integral using  different methods in \cite{Pal:2021llg,Pal:2023kgu,Levkovich-Maslyuk:2024zdy}. It would be important to clarify the relation between the two toric matrices reconstructed from the same polygonal function. This may help to determine the principle that governs the reconstruction of the full conformal integral as a linear combination of the polygonal functions associated with different basis triangles. Also, it may provide a useful example how to fix the coefficients in a linear combination of GKZ solutions representing a given Feynman integral.

A related question concerns the construction of a finite generating set for the full toric ideal. Potential guidance comes from the Yangian bootstrap for conformal integrals \cite{Loebbert:2019vcj,Levkovich-Maslyuk:2024zdy}. This approach starts from a finite set of differential constraints expressing invariance under the level-zero and level-one generators of the Yangian algebra, while the higher-level constraints follow from the Yangian relations \cite{Loebbert:2016cdm}. It would be interesting to uncover a diagrammatic interpretation of the Yangian equations and verify whether they generate the full toric ideal. Conversely, the diagrammatic algorithm may provide a natural framework for constructing Yangian invariants relevant to more general Feynman integrals with extended symmetries. The appearance of the Baxter lattice in both settings suggests a possible connection whose precise form remains to be understood.

\vspace{3mm}

\noindent \textbf{Acknowledgements.}  We are  grateful to Wladyslaw Wachowski for many useful discussions and for a series of lectures on GKZ  systems  presented at the QFT journal club at the  Lebedev Physical Institute in 2025,  as well as to all the participants for discussions during the sessions. We are also grateful to Ekaterina Mandrygina for her help in drawing diagrams. Y.Z. also thanks his dear friend Christine Batrakova for her unwavering moral support, encouragement, and inspiration throughout this work. Our work was supported by the Foundation for the Advancement of Theoretical Physics and Mathematics “BASIS”.

\appendix

\section{Lemmas on transianic indices}
\label{app:trans_rels}

\paragraph{Proofs of lemmas.} In what follows, we collect proofs of lemmas regarding transianic indices given in the main text.

\paragraph{Lemma \bref{cor1}.}
\begin{proof}
By construction, a quadratic cross-ratio $\rY_p^{\triple{ijk}}\in {}^{_1}\hspace{-0.5mm}\bm \rU_n^{\triple{ijk}}$ is represented by a quadrilateral whose four vertices consist of the three vertices of the basis triangle $\triangle_{n}^{\triple{ijk}}$ and one additional vertex $r$  (see fig.~\bref{Conf.Pol.11} ({\bf a})). According to Definition \bref{def2}, the first transianic index $b_l^{(p)}$ equals $1$  if  vertex $l$ belongs to the cross-ratio graph, and $0$ otherwise. Since exactly one vertex of the cross-ratio lies outside $\triangle_{n}^{\triple{ijk}}$, namely $r$, the first sum in \eqref{property1_U} equals $1$.

To evaluate the sum of the second transianic indices, we note that two vertices of $\triangle_{n}^{\triple{ijk}}$ are connected by a green edge; therefore, one of the indices $(b_{ij}^{(p)},b_{jk}^{(p)},b_{ki}^{(p)})$  equals $1$. The third vertex of $\triangle_{n}^{\triple{ijk}}$ is connected by a green edge to the additional vertex $r$, while the denominator is formed by the diagonals of the quadrilateral; as a consequence, exactly one edge of $\triangle_{n}^{\triple{ijk}}$ appears as an orange edge and contributes $-1$, whereas the remaining edge of $\triangle_{n}^{\triple{ijk}}$ does not belong to the cross-ratio graph and contributes $0$. Hence, the second sum in \eqref{property1_U} is  $1-1+0=0$, which proves the claim for the subset ${}^{_1}\hspace{-0.5mm}\bm \rU_n^{\triple{ijk}}$.

By construction, a quadratic cross-ratio $\rY_p^{\triple{ijk}} \in {}^{_2}\hspace{-0.5mm}\bm \rU_n^{\triple{ijk}}$ is represented by a quadrilateral whose four vertices consist of two vertices of $\triangle_{n}^{\triple{ijk}}$ and two additional vertices $r$ and $s$ (see fig.~ \bref{Conf.Pol.11} ({\bf b})). Since exactly  two vertices of the cross-ratio lie outside $\triangle_{n}^{\triple{ijk}}$, namely $r$ and $s$, the first sum in \eqref{property1_U} equals 2.

To evaluate the sum of the second transianic indices, we observe that the vertices $r$ and $s$ are chosen from the open chambers; therefore, they do not coincide with the vertices of $\triangle_{n}^{\triple{ijk}}$. Consequently, the diagonals associated with the denominator of the cross-ratio (colored orange in fig.~\bref{Conf.Pol.11}) never coincide with any edge of $\triangle_{n}^{\triple{ijk}}$. In contrast, exactly one edge of $\triangle_{n}^{\triple{ijk}}$ coincides with a green edge of the cross-ratio graph. According to Definition \bref{def2}, this edge contributes $1$ to the second sum in \eqref{property1_U}, thereby completing the proof of the first line in Lemma \bref{cor1}.

By construction, a cubic cross-ratio $\rY_p^{\triple{ijk}}\in \bm \rW_n^{\triple{ijk}}$ is represented by a pentagon whose five vertices consist of the three vertices of $\triangle_{n}^{\triple{ijk}}$ and two additional vertices $r$ and $s$ (see fig.~\bref{Conf.Pol.11} ({\bf c})). Since exactly two vertices of the cross-ratio lie outside $\triangle_{n}^{\triple{ijk}}$, namely, $r$ and $s$, the first sum in \eqref{property1_W} equals  $2$.

To evaluate the sum of the second transianic indices, we note that, by construction, two edges of $\triangle_{n}^{\triple{ijk}}$ are always green, while the remaining edge is orange. Therefore, among the three indices $(b_{ij}^{(p)},b_{jk}^{(p)},b_{ki}^{(p)})$, two are equal to $1$ and the third is equal to $-1$. Hence, the second sum in \eqref{property1_W} is $1+1-1=1$, which proves the claim for the subset $\bm \rW_n^{\triple{ijk}}$.
\end{proof}

\paragraph{Lemma \bref{lemma:local_inversion_balance}.}
\begin{proof}
Consider a cross-ratio $\rY_p^{\triple{ijk}}\in \bm \rY_n^{\triple{ijk}}$. Let $\textsf{Or}_i^{\triangle}$ and $\textsf{Gr}_i^{\triangle}$ denote the numbers of orange and green chords of its diagram that are incident to the vertex $i$ and belong to the basis triangle $\triangle_n^{\triple{ijk}}$. By Definition \bref{def2}, the second transianic index equals $1$ for a green edge, $-1$ for an orange edge, and $0$ for an absent edge. Since the two edges of the basis triangle incident to $i$ are $(i,j)$ and $(k,i)$, it follows that
\be
\label{lemma_helping_prop_1}
-b_{ij}^{(p)}-b_{ki}^{(p)}
=\textsf{Or}_i^{\triangle}-\textsf{Gr}_i^{\triangle}\,.
\ee

For $l\in\cham_n^{(i)}$, the diagrammatic algorithm ensures that the basis chord $(i,l)$ is orange precisely when $l$ belongs to the cross-ratio diagram, and is absent otherwise. By Definition \bref{def2}, these two cases correspond to $b_l^{(p)}=1$ and $b_l^{(p)}=0$, respectively. Therefore,
\be
\label{lemma_helping_prop_2}
\sum_{l\in\cham_n^{(i)}}b_l^{(p)}=\textsf{Or}_i^{\vee}\,,
\ee
where $\textsf{Or}_i^{\vee}$ denotes the number of orange basis chords $(i,l)$ with $l\in\cham_n^{(i)}$ appearing in the cross-ratio diagram.

By Proposition \bref{rem:cross_ratio_non_basis_chord}, the unique non-basis chord is green. Hence, every orange chord is a basis chord. In particular, the orange chords incident to $i$ consist of the $\textsf{Or}_i^{\triangle}$ basis-triangle edges and the $\textsf{Or}_i^{\vee}$ remaining basis chords. Combining \eqref{lemma_helping_prop_1} and
\eqref{lemma_helping_prop_2}, we thus obtain
\be
-b_{ij}^{(p)}-b_{ki}^{(p)}
+\sum_{l\in\cham_n^{(i)}}b_l^{(p)}
=\textsf{Or}_i^{\vee}+\textsf{Or}_i^{\triangle}-\textsf{Gr}_i^{\triangle}
=\textsf{Or}_i-\textsf{Gr}_i^{\triangle}\,,
\ee
where $\textsf{Or}_i$ is the total number of orange chords incident to the vertex $i$.

For both the quadratic and cubic cross-ratios \eqref{U_inv} and \eqref{W_inv}, each vertex label occurs equally many times among the endpoints of the squared distances in the numerator and denominator. Indeed, for the quadratic cross-ratio \eqref{U_inv} each label occurs once in each, while for the cubic cross-ratio \eqref{W_inv} the label $j_2$ occurs twice in each and every other label occurs once in each. Consequently, the total number $\textsf{Gr}_i$ of green chords incident to the vertex $i$ equals $\textsf{Or}_i$. This also holds when $i$ is absent from the diagram, in which case both numbers vanish. Therefore,
\be
\label{lemma_helping_prop_3}
-b_{ij}^{(p)}-b_{ki}^{(p)}
+\sum_{l\in\cham_n^{(i)}}b_l^{(p)}
=\textsf{Gr}_i-\textsf{Gr}_i^{\triangle}\,.
\ee

It remains to specify the right-hand side. According to  the diagrammatic algorithm, every green chord is either an edge of the basis triangle or the unique non-basis chord $(r_p,s_p)$ identified in Proposition \bref{rem:cross_ratio_non_basis_chord}. Consequently, $\textsf{Gr}_i-\textsf{Gr}_i^{\triangle}$ counts precisely whether this non-basis chord is incident to the vertex  $i$:
\be
\textsf{Gr}_i-\textsf{Gr}_i^{\triangle}
=\begin{cases}
1\,, & i\in\{r_p,s_p\}\,,\\
0\,, & i\notin\{r_p,s_p\}\,.
\end{cases}
\ee
Substituting this expression into
\eqref{lemma_helping_prop_3} proves
\eqref{local_inversion_balance}. The relations at $j$
and $k$ follow by cyclic permutations of $i,j,k$.
\end{proof}

\paragraph{Additional lemmas.} In what follows, we establish a few lemmas that  demonstrate additional constraints satisfied by the transianic indices.

\begin{lemma}
\label{lemma:second}
The admissible second transianic indices  for each type of cross-ratio variable are listed below. Every line gives one admissible assignment.

\vspace{4mm}
$\rY_p^{\triple{ijk}}\in\bm \rW_n^{\triple{ijk}}$:
\be
\label{group_1}
\begin{aligned}
b_{ij}^{(p)} &= 1\,,  & b_{jk}^{(p)} &= 1\,,  & b_{ki}^{(p)} &= -1\,, \\[2pt]
b_{ij}^{(p)} &= 1\,,  & b_{jk}^{(p)} &= -1\,, & b_{ki}^{(p)} &= 1\,,  \\[2pt]
b_{ij}^{(p)} &= -1\,, & b_{jk}^{(p)} &= 1\,,  & b_{ki}^{(p)} &= 1\,.
\end{aligned}
\ee

$\rY_p^{\triple{ijk}}\in{}^{_1}\hspace{-0.5mm}\bm \rU_n^{\triple{ijk}}$:
\be
\label{group_2}
\begin{aligned}
b_{ij}^{(p)} &= 1\,,  & b_{jk}^{(p)} &= -1\,, & b_{ki}^{(p)} &= 0\,, \\[2pt]
b_{ij}^{(p)} &= -1\,, & b_{jk}^{(p)} &= 1\,,  & b_{ki}^{(p)} &= 0\,, \\[2pt]
b_{ij}^{(p)} &= 1\,,  & b_{jk}^{(p)} &= 0\,,  & b_{ki}^{(p)} &= -1\,, \\[2pt]
b_{ij}^{(p)} &= -1\,, & b_{jk}^{(p)} &= 0\,,  & b_{ki}^{(p)} &= 1\,, \\[2pt]
b_{ij}^{(p)} &= 0\,,  & b_{jk}^{(p)} &= 1\,,  & b_{ki}^{(p)} &= -1\,, \\[2pt]
b_{ij}^{(p)} &= 0\,,  & b_{jk}^{(p)} &= -1\,, & b_{ki}^{(p)} &= 1\,.
\end{aligned}
\ee

$\rY_p^{\triple{ijk}}\in{}^{_2}\hspace{-0.5mm}\bm \rU_n^{\triple{ijk}}$:
\be
\label{group_3}
\begin{aligned}
b_{ij}^{(p)} &= 0\,, & b_{jk}^{(p)} &= 0\,, & b_{ki}^{(p)} &= 1\,, \\[2pt]
b_{ij}^{(p)} &= 0\,, & b_{jk}^{(p)} &= 1\,, & b_{ki}^{(p)} &= 0\,, \\[2pt]
b_{ij}^{(p)} &= 1\,, & b_{jk}^{(p)} &= 0\,, & b_{ki}^{(p)} &= 0\,.
\end{aligned}
\ee
\end{lemma}

\begin{proof}
Lemma \bref{cor3}  for  $\rY_p^{\triple{ijk}} \in \bm \rW_n^{\triple{ijk}}$ yields $b_{ij}^{(p)} + b_{jk}^{(p)} + b_{ki}^{(p)} = 1$. By the construction of a cubic cross-ratio, two edges of $\triangle_{n}^{\triple{ijk}}$ are always green and the remaining one is necessarily orange. Therefore, two of the three second transianic indices equal $1$, while the third  equals  $-1$. It follows that the only admissible assignments are precisely those listed in \eqref{group_1}.

Lemma \bref{cor1} for $\rY_p^{\triple{ijk}} \in {}^{_1}\hspace{-0.5mm}\bm \rU_n^{\triple{ijk}}$  yields $b_{ij}^{(p)} + b_{jk}^{(p)} + b_{ki}^{(p)} = 0$. By the construction of this type of cross-ratio, one edge of $\triangle_{n}^{\triple{ijk}}$ is green; hence, one of the indices $(b_{ij}^{(p)},b_{jk}^{(p)},b_{ki}^{(p)})$ equals $1$. At the same time, the diagonal construction of the denominator produces exactly one orange edge of $\triangle_{n}^{\triple{ijk}}$, so that the second index equals $-1$, while the remaining edge is not involved in the cross-ratio graph and, therefore, contributes $0$. Thus, the only admissible assignments are those listed in \eqref{group_2}.

Lemma \bref{cor1} for $\rY_p^{\triple{ijk}} \in {}^{_2}\hspace{-0.5mm}\bm \rU_n^{\triple{ijk}}$   yields  $b_{ij}^{(p)} + b_{jk}^{(p)} + b_{ki}^{(p)} = 1$. In this case, by construction, exactly one edge of $\triangle_{n}^{\triple{ijk}}$ is green, whereas the other two edges are not highlighted by either color. Consequently, exactly one of the three second transianic indices is non-vanishing and equals $1$, while the remaining two are zero. Hence, the only admissible assignments are those listed in \eqref{group_3}.

\end{proof}

\begin{figure}
    \centering
    \includegraphics[width=0.5\linewidth]{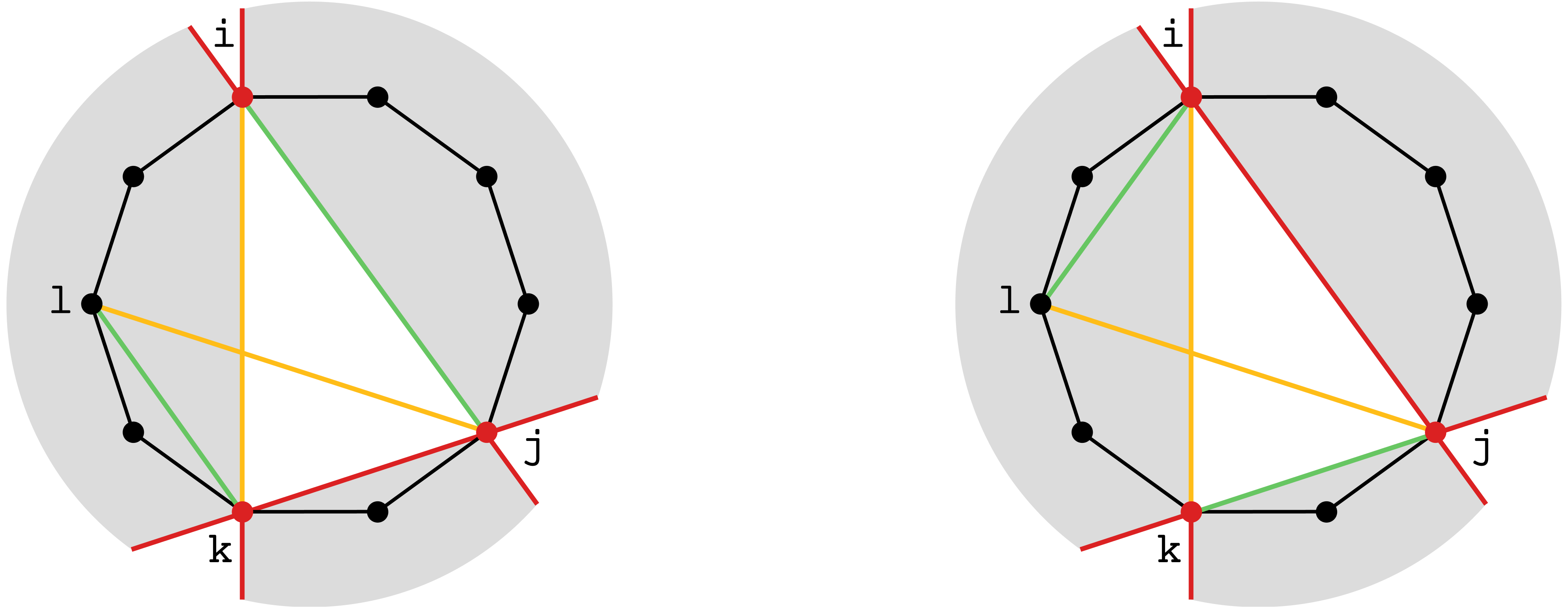}
    \caption{Configurations of ${}^{_1}\hspace{-0.5mm}\bm \rU_n^{\triple{ijk}}$ for a fixed vertex $l$.}
    \label{fan_lemma_pic_1}
\end{figure}

The identities contained  in the preceding lemmas  are column-wise relations in the transianic matrix: for a fixed cross-ratio $\rY_p^{\triple{ijk}}$, one sums the corresponding transianic index values over the vertices of the conformal polygon or over the edges of the basis triangle. The next two lemmas have a different nature. They describe row-wise relations: one fixes either a vertex $l\notin \{i,j,k\}$ or an edge of the basis triangle, and then sums the associated transianic index values over all cross-ratios $\rY_p^{\triple{ijk}}$ for $p=1,...,\xr$. Thus, while Lemmas \bref{cor1} and \bref{lemma:local_inversion_balance}, together with Lemma \bref{lemma:second}, describe the internal structure of each individual cross-ratio diagram, the following Lemmas \bref{fan_lemma_1} and \bref{fan_lemma_2} characterize how a given vertex or a basis triangle edge is distributed throughout the whole set of cross-ratio diagrams.

\begin{lemma}
\label{fan_lemma_1}
For any  $l \notin \{i,j,k\}$, the first transianic indices satisfy
\be
\sum_{p=1}^{\xr} b^{(p)}_l = n-2\,.
\ee
\end{lemma}
\begin{proof}
Let $l\in \mathbb{C}_n^{(i)}$, while the two remaining chambers are $\mathbb{C}_n^{(j)}$ and $\mathbb{C}_n^{(k)}$. We now count, separately for each subset of cross-ratios, the number of elements for which the index $b_l^{(p)}$ is non-zero.

${}^{_1}\hspace{-0.5mm}\bm \rU_n^{\triple{ijk}}$: Each cross-ratio contains exactly one vertex that does not belong to $\triangle_{n}^{\triple{ijk}}$. Therefore, for a given vertex $l$, there are precisely two admissible configurations, shown in fig.~\bref{fan_lemma_pic_1}.

${}^{_2}\hspace{-0.5mm}\bm \rU_n^{\triple{ijk}}$: Each cross-ratio contains two vertices that do not belong to $\triangle_{n}^{\triple{ijk}}$. By construction, these two vertices lie in two different open chambers. Since $l\in\mathbb{C}_n^{(i)}$, the second vertex can be chosen arbitrarily either from $\mathbb{C}_n^{(j)}$ or from $\mathbb{C}_n^{(k)}$. Therefore, for a given  $l$, the number of such cross-ratios is equal to $L_n^{(j)} + L_n^{(k)}$.

$\bm \rW_n^{\triple{ijk}}$: Each cross-ratio contains exactly two vertices outside $\triangle_{n}^{\triple{ijk}}$, and these two vertices lie in the same open chamber. Since $l\in \mathbb{C}_n^{(i)}$, the second vertex can be chosen arbitrarily among the remaining vertices of the same chamber $\mathbb{C}_n^{(i)}$. Therefore, for a given $l$, the number of such cross-ratios is equal to  $L_n^{(i)} - 1$.  Thus, the sum is given by
\be
\sum_{p=1}^{\xr} b_l^{(p)}= 2 + L_n^{(j)} + L_n^{(k)} + L_n^{(i)} - 1 = n-2\,,
\ee
where in the last step we used the balance relation \eqref{balance_relation}.
\end{proof}

\begin{lemma}
\label{fan_lemma_2}
The second transianic indices satisfy
\be
\label{equations_for_fan_lemma_2}
\sum_{p=1}^{\xr}b_{ij}^{(p)} = \frac{(n-3)(n-2)}{2} - (n-1)L_n^{(k)}\,,
\ee
where $L_n^{(k)}$ is given in  (\ref{L_in_vertices}). Two more relations are obtained by the cyclic permutations of $i,j,k$.
\end{lemma}
\begin{proof}
Let us count, separately for each subset of cross-ratios in \eqref{Y_split}, the number of elements for which the index $b_{ij}^{(p)}$ is non-zero.

${}^{_1}\hspace{-0.5mm}\bm \rU_n^{\triple{ijk}}$: First, assume that the edge $(i,j)$ is green. Then, an additional vertex $l$ can be chosen from either of the two chambers adjacent to this edge, namely, $\mathbb{C}_n^{(i)}$ or $\mathbb{C}_n^{(j)}$. Hence, there are $L_n^{(i)}+L_n^{(j)}$ possible choices for $l$, as illustrated in fig.~\bref{fan_lemma_pic_2} {\bf (a)},~{\bf (b)}. These configurations have $b_{ij}^{(p)}=1$. Now, assume that the edge $(i,j)$ is orange. Then, the vertex outside $\triangle_{n}^{\triple{ijk}}$ must lie opposite to the edge $(i,j)$, i.e., in the chamber $\mathbb{C}_n^{(k)}$. For each such vertex, there are exactly two admissible cross-ratio diagrams, shown in fig.~\bref{fan_lemma_pic_2} {\bf (c)}, {\bf (d)}. Hence, the total number of such diagrams equals  $2L_n^{(k)}$. These diagrams have $b_{ij}^{(p)}=-1$.

${}^{_2}\hspace{-0.5mm}\bm \rU_n^{\triple{ijk}}$: The edge $(i,j)$ cannot be orange. If the edge $(i,j)$ is green, then the two external vertices must be chosen from the adjacent chambers $\mathbb{C}_n^{(i)}$ and $\mathbb{C}_n^{(j)}$. Therefore, the total number of such diagrams equals  $L_n^{(i)} L_n^{(j)}$, and each contributes $+1$ to the sum, see fig.~\bref{fan_lemma_pic_2} {\bf (e)}.

$\bm \rW_n^{\triple{ijk}}$: For the diagrams in this subset, two edges of $\triangle_{n}^{\triple{ijk}}$ are always green, while the remaining one is orange. The edge $(i,j)$ is  green if the two external vertices of the corresponding cross-ratio are chosen either from the chamber $\mathbb{C}_n^{(i)}$ or from the chamber $\mathbb{C}_n^{(j)}$. This yields the positive contributions $\binom{L_n^{(i)}}{2}$  and $\binom{L_n^{(j)}}{2}$. On the other hand, the edge $(i,j)$ is orange if the two external vertices are chosen from the opposite chamber $\mathbb{C}_n^{(k)}$. This yields  the negative contribution $\binom{L_n^{(k)}}{2}$.
All in all, the required sum is given by
\be
\ba{c}
\dps \sum_{p=1}^{\xr} b_{ij}^{(p)} = L_n^{(i)}+L_n^{(j)} - 2L_n^{(k)} {+L_n^{(i)}L_n^{(j)}} + \binom{L_n^{(i)}}{2} + \binom{L_n^{(j)}}{2} - \binom{L_n^{(k)}}{2}
\\
\dps= \frac{(n-3)(n-2)}{2} - (n-1)L_n^{(k)}\,.
\ea
\ee
\end{proof}
\begin{figure}
    \centering
    \includegraphics[width=1\linewidth]{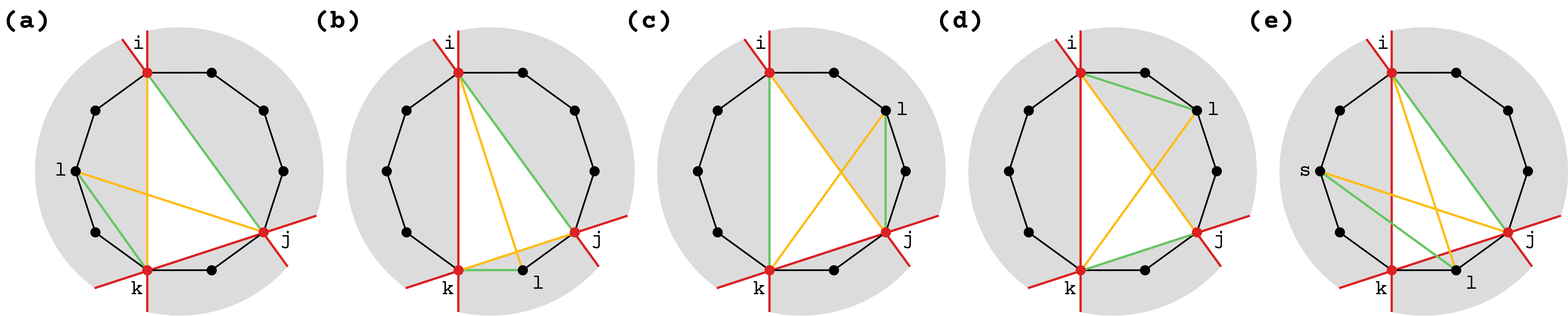}
    \caption{{\bf (a)},{\bf (b)}: Configurations of ${}^{_1}\hspace{-0.5mm}\bm \rU_n^{\triple{ijk}}$ for the green edge $(i,j)$. {\bf (c)},{\bf (d)}: Configurations of ${}^{_1}\hspace{-0.5mm}\bm \rU_n^{\triple{ijk}}$ for the orange edge  $(i,j)$. {\bf (e)}: Configurations of ${}^{_2}\hspace{-0.5mm}\bm \rU_n^{\triple{ijk}}$ for the green edge  $(i,j)$.}
    \label{fan_lemma_pic_2}
\end{figure}


\section{Examples}
\label{sec:examples}

In this section, we illustrate two complementary ways of constructing the differential systems associated with  low-point polygonal functions. We begin with the Appell $F_4$ function, which arises as the $n=4$ polygonal function, and derive its GKZ system independently of the general results of the main text. Starting from the series representation, we extract the GKZ data, formulate the corresponding GKZ equations, and show that they reduce to the standard differential system for $F_4$. We then apply Corollary \bref{cor:finite_toric_subsystem} and Theorem \bref{proposition_for_Euler_eq} directly to the $n=5,6$ polygonal functions, which are Srivastava--Daoust functions, and write their toric and Euler equations. For $n=5,6$, we further illustrate Remark \bref{rem:diagrammatic_subsystem} by providing toric equations that cannot be obtained from the finite subsystems associated with the respective lattice bases by either algebraic combination or differentiation.

\subsection{GKZ representation for the Appell $F_4$ function}
\label{sec:Appell}
Consider the fourth Appell function:
\be
\label{Appell}
F_4 \bigg[
\begin{array}{ll}
a, b
\vspace{-3mm}
\\
c, d
\end{array}\Big|
x, y \bigg]
= \sum_{m,n=0}^\infty\frac{(a)_{m+n}(b)_{m+n}}{(c)_m(d)_n}\frac{x^m}{m!}\frac{y^n}{n!}\,,
\ee
where $a,b,c,d\in\mathbb{C}$, with $c,d\notin\mathbb{Z}_{\leq 0}$, and the convergence domain is given by $\sqrt{|x|}+\sqrt{|y|}<1$.

The Euler reflection formula $\Gamma(z)\Gamma(1-z)=\pi/\sin(\pi z)$ allows us to recast the Gamma functions in the numerator and denominator into a purely reciprocal Gamma-product form. Up to a constant prefactor, we obtain
\be
\label{F_4_prop}
F_4 \propto \sum_{m,n = 0}^{\infty} \frac{1}{\Gamma(-m-n-a+1)\Gamma(-m-n-b+1)\Gamma(m+c)\Gamma(n+d)}\frac{x^m}{\Gamma(m+1)}\frac{y^n}{\Gamma(n+1)}\,.
\ee
Comparing the coefficients in this representation with those in \eqref{eq:gamma-series-gale-dual}, one deduces the GKZ data in a few steps. First, the dimensions of the relevant GKZ matrices are $N=6$, $n=4$, $\xr=2$. In particular, the Gale dual \eqref{Gale_dual_first} has the block form
\be
\cG=
\begin{pmatrix}
-\cB \\[2pt]
\phantom{-}\mathbb{1}_2
\end{pmatrix}\,,
\qquad
\text{where}
\qquad
\cB =
\begin{pmatrix}
\phantom{-}1 & \phantom{-}1 \\
\phantom{-}1 & \phantom{-}1 \\
-1 & \phantom{-}0 \\
\phantom{-}0 & -1
\end{pmatrix}\,.
\ee
The columns
\be
\label{l1l2}
{\bm \cG}^{(1)} = (-1,-1,1,0,1,0)^\rT\,,
\qquad
{\bm \cG}^{(2)} = (-1,-1,0,1,0,1)^\rT\, ,
\ee
form a basis of the lattice of relations $\mathbb{L}={\rm im}_{\mathbb{Z}}\,\cG=\left\{\cG{\bm m}\mid {\bm m} = (m,n) \in \mathbb{Z}^2\right \}$. The balance constraint \eqref{linear} is obviously satisfied.

The toric matrix $\cA$ is found as a solution to $\cA\cG=0$. Clearly, any matrix of the form $\cA = U \left(\mathbb{1}_4 \mid \cB \right)$, where $U \in{\rm GL}(4,\mathbb{Z})$, satisfies this condition:
\be
\cA\cG = U\left(\mathbb{1}_4 \mid \cB \right)\begin{pmatrix}
-\cB \\[2pt]
\mathbb{1}_2
\end{pmatrix} =  U(-\cB +\cB) = 0\,.
\ee
We adopt the convenient normalization in which $U=\mathbb{1}_4$. Then, the $4\times6$ toric matrix is given by
\be
\cA =
\begin{pmatrix*}[r]
1 & 0 & 0 & 0 & 1 & 1 \\
0 & 1 & 0 & 0 & 1 & 1 \\
0 & 0 & 1 & 0 & -1 & 0 \\
0 & 0 & 0 & 1 & 0 & -1
\end{pmatrix*}.
\ee

The denominator in \eqref{F_4_prop} also determines the parameter vector:
\be
\bm \gamma=(-a,-b,c-1,d-1,0,0)\in \mathbb{C}^6\,.
\ee
Hence, using the relation  \eqref{alpha_A_gamma}, we find that the parameter vector $\bm \alpha$ is given by
\be
\bm \alpha = (-a,-b,c-1,d-1) \in \mathbb{C}^4\,.
\ee

Note that the two vectors have the same non-vanishing components. This agreement  is a direct consequence of the above choice of $\cA$. Since the left block of $\cA$ is the identity matrix and the last two components of $\bm\gamma$ vanish, the relation $\bm\alpha = \cA\bm\gamma$ simply reproduces the first four components of $\bm\gamma$. Hence, the equality of the non-vanishing entries of $\bm\alpha$ and $\bm\gamma$ is a result of this particular normalization of $\cA$, rather than an invariant property of the GKZ data. While $\bm\alpha$ is uniquely determined by $\bm\gamma$ through $\bm\alpha = \cA\bm\gamma$, the converse is not generally true. For fixed $\bm\alpha$ and $\cA$, equation \eqref{alpha_A_gamma} admits  multiple solutions for $\bm\gamma$, since $\cA$ is typically rectangular and the system is underdetermined. Any two such solutions differ by an element of $\ker\cA$.

At this stage, all the GKZ data have been extracted from the Appell $F_4$ function \eqref{Appell}. Now, we use them to formulate the GKZ system following Definition \bref{prop:GKZ} and Proposition \bref{prop:solution}. Denoting the derivatives $\partial_i=\partial/ \partial v_i$, $i=1,...,6$, we express the  toric   \eqref{toric_eq} and Euler \eqref{Euler_eq} equations as
\be
\label{toric_F_4}
\ba{c}
\hspace{-8mm}{\bm \cG}^{(1)}:\quad \partial_3\partial_5\Phi - \partial_1\partial_2\Phi = 0\,,
\\
\hspace{-8mm}{\bm \cG}^{(2)}:\quad  \partial_4\partial_6\Phi - \partial_1\partial_2\Phi = 0\,,
\ea
\ee
where ${\bm \cG}^{(1)}$ and ${\bm \cG}^{(2)}$ are given by \eqref{l1l2}. The Euler equations are
\be
\label{Euler_F_4}
\begin{array}{rl}
& (\theta_1 + \theta_5 +\theta_6 + a)\Phi = 0\, , \\
& (\theta_2 + \theta_5 + \theta_6 + b)\Phi = 0\, , \\
& (\theta_3 -\theta_5 -c +1)\Phi = 0\, , \\
& (\theta_4 - \theta_6 -d + 1)\Phi = 0\, ,
\end{array}
\ee
where $\Phi = \Phi(v_1,...,v_6)$ is the corresponding $\Gamma$-series:
\be
\label{step1}
\Phi_{\cG,\bm \gamma} = \sum_{m,n = 0}^\infty \frac{v_1^{-m-n-a}}{\Gamma(-m-n-a+1)} \frac{v_2^{-m-n-b}}{\Gamma(-m-n-b+1)} \frac{v_3^{m+c-1}}{\Gamma(m+c)} \frac{v_4^{n+d-1}}{
\Gamma(n+d)}\frac{v_5^{m}}{\Gamma(m+1)}\frac{v_6^{n}}{\Gamma(n+1)}\,.
\ee
Using the Euler reflection formula in reverse, we represent this series, up to a constant prefactor, as follows:
\be
\label{step2}
\Phi_{\cG,\bm \gamma} \propto  v_1^{-a}v_2^{-b}v_3^{c-1}v_4^{d-1}\sum_{m,n = 0}^\infty  \frac{\Gamma(m+n+a)\Gamma(m+n+b)}{
\Gamma(m+c)\Gamma(n+d)\Gamma(m+1)\Gamma(n+1)}\left[\frac{v_3v_5}{v_1v_2}\right]^m\left[\frac{v_4v_6}{v_1v_2}\right]^n\,.
\ee
Thus, recalling the original Appell $F_4$ series \eqref{Appell}, we conclude that its GKZ representation has the form
\be
\label{Phi_is_F}
\Phi_{\cG,\bm \gamma} \propto v_1^{-a}v_2^{-b}v_3^{c-1}v_4^{d-1}F_4
\bigg[
\begin{array}{ll}
a, b
\vspace{-3mm}
\\
c, d
\end{array}\Big|
\frac{v_3v_5}{v_1v_2},\frac{v_4v_6}{v_1v_2} \bigg]\,,
\ee
where the two original variables $\{x,y\}$ are now given by homogeneous combinations of the six GKZ  variables $\{v_1, ..., v_6\}$:
\be
x= \frac{v_3v_5}{v_1v_2}\,,
\qquad
y = \frac{v_4v_6}{v_1v_2}\,.
\ee

Finally, substituting \eqref{Phi_is_F} into the toric equations \eqref{toric_F_4}, and noting that the Euler equations \eqref{Euler_F_4} have already been solved, we obtain two partial differential equations in $x$ and $y$:
\be
\label{system_for_F_4}
\ba{l}
\left(x(1-x)\partial_x^2 - 2xy \partial_x\partial_y -y^2 \partial_y^2 +(c-x(a+b+1))\partial_x -y(a+b+1)\partial_y -ab\right) F_4 = 0\,,
\vspace{2mm}
\\
\left(y(1-y)\partial_y^2 - 2xy \partial_x\partial_y -x^2 \partial_x^2 +(d-y(a+b+1))\partial_y - x(a+b+1)\partial_x-ab\right) F_4 = 0\,.
\ea
\ee
This is the standard PDE system for the Appell $F_4$ function  \cite{Bateman:100233}.

\subsection{Pentagonal function}
Consider the $n=5$ basis function \cite{Alkalaev:2025fgn,Alkalaev:2025zhg}:
\be
\label{master_5}
\BF_5^{\triple{123}}({\bm a}|\bm x)
\; \propto \;
\rV_5^{\triple{123}}(\bm a|\bm x)\,
\HG_5^{\triple{123}}\left({\bm a}\big|\rY^{\triple{123} }_1\,, ...\,, \rY^{\triple{123}}_5 \right)\,.
\ee
The diagrammatic algorithm constructs the leg-factor \eqref{V_n}:
\be
\label{V_5}
\rV_5^{\triple{123}}(\bm a|\bm x)=
X_{12}^{|\bm a_{1,2}|'} X_{23}^{|\bm a_{2,3}|'}  X_{31}^{|\bm a_{3,1}|'}\, X_{24}^{-a_4} X_{25}^{-a_5}\,,
\ee
and the set of cross-ratios \eqref{DA_cr}:
\begin{equation}
\label{cross_5_123}
\begin{array}{l}
\dps \rY^{\triple{123}}_1 = \frac{X_{12} X_{34} }{X_{13} X_{24}}\,,
\qquad \rY^{\triple{123}}_2 = \frac{X_{23} X_{14} }{X_{13} X_{24}}\,,
\qquad \rY^{\triple{123}}_3 = \frac{X_{12} X_{35} }{X_{13} X_{25}}\,,
\\[15pt]
\dps \rY^{\triple{123}}_4 = \frac{X_{23} X_{15} }{X_{13} X_{25}}\,, \qquad \rY^{\triple{123}}_5 = \frac{X_{12} X_{23} X_{45} }{X_{13} X_{24} X_{25}}\,.
\end{array}
\end{equation}
Under the variable identification established in Proposition \bref{prop:mapping_new}, equation \eqref{Gamma_series_with_F_V_H} yields
\be
\label{fin_fun}
\ba{l}
\dps
\Phi_{\cG,\bm\gamma}(\bm v)
=
-\frac{1}{\pi^5}
\bigg[\prod_{l=1}^{5}\Gamma(a_l)\bigg]
\bigg[\prod_{l=4}^{5}\sin(\pi a_l)\bigg]
\vspace{2mm}
\\
\dps
\hspace{12mm}\times\;
\sin\!\big(\pi|{\bm a}_{1,2}|'\big)
\sin\!\big(\pi|{\bm a}_{2,3}|'\big)
\sin\!\big(\pi|{\bm a}_{3,1}|'\big)\,
\BF_5^{\triple{123}}({\bm a}|\bm x)\,.
\ea
\ee
The factor relating the two functions is $\bm x$-independent; therefore, the $\Gamma$-series and the basis function satisfy the same GKZ system.

Denote $\partial_{ab}=\partial/\partial X_{ab}$ and $\theta_{ab}=X_{ab}\partial_{ab}$. Given the set \eqref{cross_5_123}, Theorem \bref{proposition_for_Euler_eq} yields the system of $5$ Euler equations:
\be
\label{51systems_euler}
\begin{array}{rl}
&(\theta_{12}-\theta_{34}-\theta_{35}-\theta_{45}-|\bm a_{1,2}|')\BF_5^{\triple{123}}=0\,,
\\
&(\theta_{23}-\theta_{14}-\theta_{15}-\theta_{45}-|\bm a_{2,3}|')\BF_5^{\triple{123}}=0\,,
\\
&(\theta_{31}+\theta_{14}+\theta_{15}+\theta_{34}+\theta_{35}+\theta_{45}-|\bm a_{3,1}|')\BF_5^{\triple{123}}=0\,,
\\
&(\theta_{24}+\theta_{14}+\theta_{34}+\theta_{45}+a_4)\BF_5^{\triple{123}}=0\,,
\\
&(\theta_{25}+\theta_{15}+\theta_{35}+\theta_{45}+a_5)\BF_5^{\triple{123}}=0\,,
\end{array}
\ee
while Corollary \bref{cor:finite_toric_subsystem} yields a finite subsystem of $5$ toric equations:
\be
\label{51systems}
\begin{array}{rl}
&(\partial_{23}\partial_{15}
-\partial_{31}\partial_{25})\BF_5^{\triple{123}}=0\,, \qquad (\partial_{12}\partial_{34}
-\partial_{31}\partial_{24})\BF_5^{\triple{123}}=0\,,
\\
& (\partial_{23}\partial_{14}
-\partial_{31}\partial_{24})\BF_5^{\triple{123}}=0\,, \qquad (\partial_{12}\partial_{35}
-\partial_{31}\partial_{25})\BF_5^{\triple{123}}=0\,,
\\
& (\partial_{12}\partial_{23}\partial_{45}
-\partial_{31}\partial_{24}\partial_{25})\BF_5^{\triple{123}}=0\,.
\end{array}
\ee
Here, the four second-order toric equations arise from the quadratic cross-ratios $\bm \rU_5^{\triple{123}} = \{\rY_1^{\triple{123}},...,\rY_4^{\triple{123}}\}$, whereas the third-order equation arises from the cubic cross-ratio $\mathbf{W}_5^{\triple{123}} = \{\rY_5^{\triple{123}}\}$, cf. Remark \bref{remark_tells_the_order}.

For $\bm m=(-1,1,1,-1,0)\in\mathbb{Z}^{5}$, Theorem \bref{prop:general_toric_equation} yields the toric equation
\be
\label{pentagon_additional_toric_equation}
\big(\partial_{14}\partial_{35}
-\partial_{15}\partial_{34}\big)
\BF_5^{\triple{123}}=0\,,
\ee
which is neither a differential nor an algebraic consequence of the finite subsystem \eqref{51systems}. Thus, it illustrates the fact that the toric equations associated with the chosen basis of the lattice do not generate all toric equations, as stated in Remark \bref{rem:diagrammatic_subsystem}.
\subsection{Hexagonal function}
\label{sec:SD_funct}
Consider the $n=6$ basis function \cite{Alkalaev:2025fgn,Alkalaev:2025zhg}:
\be
\label{master_6}
\BF_6^{\triple{123}}({\bm a}|\bm x)
\; \propto \;
\rV_6^{\triple{123}}(\bm a|\bm x)\,
\HG_6^{\triple{123}}\left({\bm a}\big|\rY^{\triple{123} }_1\,, ...\,, \rY^{\triple{123}}_9 \right)\,.
\ee
The diagrammatic algorithm constructs the leg-factor \eqref{V_n}:
\be
\label{V_6_123}
\rV_6^{\triple{123}}(\bm a|\bm x)
=
X_{12}^{|\bm a_{1,2}|'} X_{23}^{|\bm a_{2,3}|'}  X_{31}^{|\bm a_{3,1}|'}\, X_{24}^{-a_4} X_{25}^{-a_5} X_{26}^{-a_6}\,,
\ee
and the set of cross-ratios \eqref{DA_cr}:
\be
\label{cross_6_123}
\begin{split}
\rY^{\triple{123}}_1 &= \frac{X_{12} X_{34} }{X_{13} X_{24}}\,, \quad
\rY^{\triple{123}}_2 =  \frac{X_{23} X_{14} }{X_{13} X_{24}}\,, \qquad
\, \, \,
\rY^{\triple{123}}_3 =  \frac{X_{12} X_{35} }{X_{13} X_{25}}\,, \\
\rY^{\triple{123}}_4 &= \frac{X_{23} X_{15} }{X_{13} X_{25}}\,,  \quad
\rY^{\triple{123}}_5 =  \frac{X_{12} X_{23} X_{45} }{X_{13} X_{24} X_{25}}\,, \quad
\rY^{\triple{123}}_6 =  \frac{X_{12} X_{36} }{X_{13} X_{26} }\,, \\
\rY^{\triple{123}}_7 &= \frac{X_{23} X_{16} }{X_{13} X_{26}}\,,  \quad
\rY^{\triple{123}}_8 =  \frac{X_{12} X_{23} X_{46} }{X_{13} X_{24} X_{26}} \,, \quad
\rY^{\triple{123}}_9 =  \frac{X_{12} X_{23} X_{56} }{X_{13} X_{25} X_{26}} \,.
\end{split}
\ee
Under the variable identification established in Proposition \bref{prop:mapping_new}, equation \eqref{Gamma_series_with_F_V_H} yields
\be
\label{fin_fun_11111}
\ba{l}
\dps
\Phi_{\cG,\bm\gamma}(\bm v)
=
-\frac{1}{\pi^6}
\bigg[\prod_{l=1}^{6}\Gamma(a_l)\bigg]
\bigg[\prod_{l=4}^{6}\sin(\pi a_l)\bigg]
\vspace{2mm}
\\
\dps
\hspace{12mm}\times\;
\sin\!\big(\pi|{\bm a}_{1,2}|'\big)
\sin\!\big(\pi|{\bm a}_{2,3}|'\big)
\sin\!\big(\pi|{\bm a}_{3,1}|'\big)\,
\BF_6^{\triple{123}}({\bm a}|\bm x)\,.
\ea
\ee
The factor relating the two functions is $\bm x$-independent; therefore, the $\Gamma$-series and the basis function satisfy the same GKZ system.

Given the set \eqref{cross_6_123}, Theorem \bref{proposition_for_Euler_eq} yields the system of $6$ Euler equations:
\be
\label{61systems_euler}
\begin{array}{rl}
&  (\theta_{12}-\theta_{34}-\theta_{35}-\theta_{36}-\theta_{45}-\theta_{46}-\theta_{56}-|\bm a_{1,2}|')\BF_6^{\triple{123}}=0 \,,\\
&  (\theta_{23}-\theta_{14}-\theta_{15}-\theta_{16}-\theta_{45}-\theta_{46}-\theta_{56}-|\bm a_{2,3}|')\BF_6^{\triple{123}}=0\,, \\
&  (\theta_{31}+\theta_{14}+\theta_{15}+\theta_{16}+\theta_{34}+\theta_{35}+\theta_{36}+\theta_{45}+\theta_{46}+\theta_{56}-|\bm a_{3,1}|')\BF_6^{\triple{123}}=0\,, \\
& (\theta_{24}+\theta_{14}+\theta_{34}+\theta_{45}+\theta_{46}+a_4)\BF_6^{\triple{123}}=0 \,, \\
& (\theta_{25}+\theta_{15}+\theta_{35}+\theta_{45}+\theta_{56}+a_5)\BF_6^{\triple{123}}=0 \,, \\
& (\theta_{26}+\theta_{16}+\theta_{36}+\theta_{46}+\theta_{56}+a_6)\BF_6^{\triple{123}}=0 \,.
\end{array}
\ee
while Corollary \bref{cor:finite_toric_subsystem} yields a finite subsystem of $9$ toric equations:
\be
\label{61systems}
\begin{array}{rl}
& (\partial_{12}\partial_{23}\partial_{45} - \partial_{31}\partial_{24}\partial_{25})\BF_6^{\triple{123}}=0\,, \quad (\partial_{12}\partial_{23}\partial_{46} - \partial_{31}\partial_{24}\partial_{26})\BF_6^{\triple{123}}=0\,,\\
& (\partial_{12}\partial_{23}\partial_{56} - \partial_{31}\partial_{25}\partial_{26})\BF_6^{\triple{123}}=0\,, \quad (\partial_{12}\partial_{34}-\partial_{31}\partial_{24})\BF_6^{\triple{123}}=0\,,\\
& (\partial_{23}\partial_{14}-\partial_{31}\partial_{24})\BF_6^{\triple{123}}=0\,, \hspace{6.4mm}  \quad \quad (\partial_{12}\partial_{35}-\partial_{31}\partial_{25})\BF_6^{\triple{123}}=0\,, \\
& (\partial_{23}\partial_{15}-\partial_{31}\partial_{25})\BF_6^{\triple{123}}=0\,, \hspace{6.4mm}  \quad \quad (\partial_{12}\partial_{36}-\partial_{31}\partial_{26})\BF_6^{\triple{123}}=0\,,\\
&  (\partial_{23}\partial_{16}-\partial_{31}\partial_{26})\BF_6^{\triple{123}}=0\,.
\end{array}
\ee
Here, the six second-order toric equations arise from the quadratic cross-ratios $\bm \rU_6^{\triple{123}} = \{\rY_1^{\triple{123}},...,\rY_4^{\triple{123}},\rY_6^{\triple{123}},\rY_7^{\triple{123}}\}$, whereas the three third-order equations arise from the cubic cross-ratios $\mathbf{W}_6^{\triple{123}} = \{\rY_5^{\triple{123}},\rY_8^{\triple{123}},\rY_9^{\triple{123}}\}$, cf. Remark \bref{remark_tells_the_order}.

For $\bm m=(-1,1,0,0,0,1,-1,0,0)\in\mathbb{Z}^{9}$, Theorem \bref{prop:general_toric_equation} yields the toric equation
\be
\label{hexagon_additional_toric_equation}
\big(\partial_{14}\partial_{36}
-\partial_{16}\partial_{34}\big)
\BF_6^{\triple{123}}=0\,,
\ee
which is neither a differential nor an algebraic consequence of the finite subsystem \eqref{61systems}. Thus, it provides another illustration of the fact that the toric equations associated with the chosen basis of the lattice do not generate all toric equations, as stated in Remark \bref{rem:diagrammatic_subsystem}.

The basis function $\BF_6^{\triple{123}}$ enters the decomposition of the six-point conformal integral \cite{Alkalaev:2025fgn,Alkalaev:2025zhg}, which, in turn, was shown in \cite{Loebbert:2019vcj} to satisfy a system of $15$ second-order PDEs implied by Yangian invariance. It would be important to determine whether this finite system of Yangian equations generates the full GKZ system.

\section{Transianic indices via adjacency matrices}
\label{app:adjacency}

In this appendix, we express the transianic indices from Definition \bref{def2} in terms of the adjacency matrices of the three graphs \eqref{graph_polygon}. We first express  the second indices via  the matrix entries and then derive the first indices from the degrees of vertices.
\paragraph{Adjacency matrices.}
For a graph $G=(V,E)$ with $V\subseteq V_{\hexagon}$, we define  its adjacency matrix on the common vertex set $V_{\hexagon}=\{1,...,n\}$ by
\be
\label{adjacency_matrix_definition}
M_{rs}=
\begin{cases}
1, \quad \text{if} \quad (r,s)\in E\,,\\
0, \quad \text{otherwise}\,,
\end{cases}
\qquad r,s=1,...,n\,.
\ee
Consequently, the rows and columns corresponding to vertices from  $V_{\hexagon}\setminus V$ consist entirely of zeros. Let $P$ and $T$ denote the adjacency matrices associated with $G_{\hexagon}$ and $G_{\triangle}$, respectively. For the two-colored graph $G_{\times}$, we preserve  the signs assigned to its edges by introducing the signed adjacency matrix:
\be
\label{adjacency_cross_ratio_matrix}
Q_{rs}=
\begin{cases}
\phantom{-}1, \quad \text{if} \quad (r,s)\in E_{\textcolor{green}{\bm\times}}\,,\\
-1, \quad \text{if} \quad (r,s)\in E_{\textcolor{orange}{\bm\times}}\,,\\
\phantom{-}0, \quad \text{otherwise}\,,
\end{cases}
\qquad r,s=1,...,n\,.
\ee
All three matrices belong to $\operatorname{Mat}_{n\times n}(\mathbb{Z})$. They are symmetric and have zero diagonals, since the edges connect unordered pairs of distinct vertices. The matrix $Q$ is associated with a fixed quadratic or cubic cross-ratio.
\paragraph{Second transianic indices.}
Comparing \eqref{adjacency_cross_ratio_matrix} with \eqref{transianic}, we immediately obtain
\be
\label{adjacency_second_indices}
b_{rs}=Q_{rs}\,,\qquad (r,s)\in E_{\triangle}\,.
\ee
To extend this relation to all pairs of vertices, we set the second indices to zero for pairs of vertices that do not belong to $E_{\triangle}$:
\be
\label{adjacency_second_extension}
\widetilde b_{rs}:=
\begin{cases}
b_{rs}, \quad \text{if} \quad (r,s)\in E_{\triangle}\,,\\
0, \quad \text{otherwise}\,.
\end{cases}
\ee
Since $T_{rs}=1$ precisely on the edges of the basis triangle, this extension is given by
\be
\label{adjacency_second_full}
\widetilde b_{rs}=T_{rs}Q_{rs}\,,\qquad r,s=1,...,n\,.
\ee
The right-hand side is given by the Hadamard (entrywise) product of the adjacency matrices.
\paragraph{First transianic indices.}
Each vertex in	 $V_{\times}$ is incident to at least one edge in \eqref{cross_ratio_edge_sets}. Therefore, its presence in the cross-ratio diagram can be determined by counting its incident edges. Since every nonzero entry of $Q$ is $+1$ or $-1$, and $Q$ is symmetric, this number is
\be
\label{adjacency_vertex_degree}
d_l:=\sum_{s=1}^{n}Q_{ls}^{\,2}
=\sum_{s=1}^{n}Q_{ls}Q_{sl}
=(Q^2)_{ll}\,,\qquad l=1,...,n\,.
\ee
Here, $d_l$ is the degree of vertex $l$ in the cross-ratio graph, with $d_l=0$ when $l\notin V_{\times}$.

For a quadratic cross-ratio $U[i_1,i_2,i_3,i_4]$, the edge sets \eqref{cross_ratio_edge_sets} show that every vertex in $V_{\times}$ is incident to one green and one orange edge. Thus, $d_l=2$ for $l\in V_{\times}$ and $d_l=0$ otherwise. The first index is consequently $b_l=d_l/2$ for $l\in V_{\hexagon}\setminus V_{\triangle}$.

For a cubic cross-ratio $W[j_1,j_2,j_3,j_4,j_5]$, the vertex $j_2$ is incident to two green and two orange edges, while each of the remaining four vertices is incident  to one edge of each color. Hence,
\be
\label{adjacency_cubic_degrees}
d_l=
\begin{cases}
0, \quad \text{if} \quad l\notin V_{\times}\,,\\
2, \quad \text{if} \quad l\in\{j_1,j_3,j_4,j_5\}\,,\\
4, \quad \text{if} \quad l=j_2\,.
\end{cases}
\ee
The first index must be equal to $1$ at every participating vertex outside $V_{\triangle}$, irrespective of whether its degree is $2$ or $4$. Therefore, we seek a polynomial $f(d)=a d+c d^2$ satisfying $f(0)=0$ and $f(2)=f(4)=1$. The last two conditions give $2a+4c=1$ and $4a+16c=1$, whose solution is $a=3/4$ and $c=-1/8$. Thus, $f(d)=d(6-d)/8$, and both cases can be written uniformly as
\be
\label{adjacency_first_indices}
b_l=
\begin{cases}
\dps\frac12(Q^2)_{ll}\,,
&\text{for }U\,,\\[2mm]
\dps\frac18(Q^2)_{ll}\Big(6-(Q^2)_{ll}\Big)\,,
&\text{for }W\,,
\end{cases}
\qquad l\in V_{\hexagon}\setminus V_{\triangle}\,.
\ee
For $d_l=0,2,4$, the cubic formula yields  $b_l=0,1,1$, respectively.

It remains to express the restriction to vertices outside the basis triangle through $P$ and $T$. Every vertex of $G_{\hexagon}$ has degree two, whereas the nonzero rows of $T$ correspond precisely to the three vertices of $G_{\triangle}$. Applying the same calculation as in \eqref{adjacency_vertex_degree}, we obtain
\be
\label{adjacency_polygon_triangle_degrees}
(P^2)_{ll}=2\,,\qquad
(T^2)_{ll}=
\begin{cases}
2, \quad \text{if} \quad l\in V_{\triangle}\,,\\
0, \quad \text{otherwise}\,.
\end{cases}
\ee
Consequently, the factor $\big((P^2)_{ll}-(T^2)_{ll}\big)/2$ equals $1$ outside $V_{\triangle}$ and vanishes on $V_{\triangle}$. Extending the first indices by zero on the basis triangle,
\be
\label{adjacency_first_extension}
\widetilde b_l:=
\begin{cases}
b_l, \quad \text{if} \quad l\in V_{\hexagon}\setminus V_{\triangle}\,,\\
0, \quad \text{if} \quad l\in V_{\triangle}\,,
\end{cases}
\ee
and multiplying \eqref{adjacency_first_indices} by this factor yields, for $l=1,...,n$,
\be
\label{adjacency_first_full}
\widetilde b_l=
\begin{cases}
\dps\frac14\Big((P^2)_{ll}-(T^2)_{ll}\Big)(Q^2)_{ll}\,,
&\text{for }U\,,\\[3mm]
\dps\frac1{16}\Big((P^2)_{ll}-(T^2)_{ll}\Big)
(Q^2)_{ll}\Big(6-(Q^2)_{ll}\Big)\,,
&\text{for }W\,.
\end{cases}
\ee
Equations \eqref{adjacency_second_full} and \eqref{adjacency_first_full} give the desired expressions in terms of  the three adjacency matrices. Their restrictions to $(r,s)\in E_{\triangle}$ and $l\in V_{\hexagon}\setminus V_{\triangle}$ reproduce the original transianic indices.

For the cross-ratios selected by the diagrammatic algorithm, the formula for the first indices simplifies further. In every cubic diagram $W[i,j,k,r,s]$, including the cyclic permutations of $i,j,k$, the vertex of degree four belongs to $V_{\triangle}$. Hence, every participating vertex outside  the basis triangle has degree two for both quadratic and cubic cross-ratios. On this set of diagrams, \eqref{adjacency_first_full} reduces to the single expression
\be
\label{adjacency_first_algorithm}
\widetilde b_l=\frac14\Big((P^2)_{ll}-(T^2)_{ll}\Big)(Q^2)_{ll}\,,
\qquad l=1,...,n\,.
\ee

The expressions \eqref{adjacency_first_indices} and \eqref{adjacency_first_full} apply to arbitrary quadratic and cubic cross-ratio diagrams relative to a fixed basis triangle. Their derivation does not use the diagrammatic algorithm. Instead, for the cross-ratios selected by this algorithm, the first indices admit a simpler expression through individual entries of the signed adjacency matrix. Indeed, for $l\in\cham_n^{(a)}$, the vertex $l$ belongs to the cross-ratio diagram precisely when the basis chord $(a,l)$ occurs in it; this chord is then orange. Thus, taking $Q$ to be the signed adjacency matrix of $\rY_p^{\triple{ijk}}$, we obtain
\be
\label{adjacency_first_Gale}
b_l^{(p)} = -Q_{al}^{(p)} = -\cG_{(a,l)}^{(p)}\,, \qquad a\in V_{\triangle}\,, \quad l\in\cham_n^{(a)}\,,
\ee
where the last equality follows from \eqref{Gale_matrix_chord_entries}. Thus, for a fixed  $l$, the sequence $(b_l^{(1)},...,b_l^{(\xr)})$  coincides with the row of the Gale dual corresponding to $(a,l)$, with its sign reversed.


\providecommand{\href}[2]{#2}\begingroup\raggedright\endgroup

\end{document}